\documentclass[a4paper,UKenglish,cleveref, autoref, thm-restate]{lipics-v2021}

\title{Counting Patterns in Degenerate Graphs in Constant Space} 

\author{Balagopal Komarath}{Department of Computer Science,
IIT Gandhinagar, India}{bkomarath@iitgn.ac.in}{}{}
\author{Anant Kumar}{Department of Computer Science,
IIT Gandhinagar, India}{kumar\_anant@iitgn.ac.in}{}{}

\author{Akash Pareek}{Department of Computer Science,
Université Libre de Bruxelles, Belgium}{akash.pareek@ulb.be}{}{}

\authorrunning{B. Komarath, A. Kumar, and A. Pareek} 

\Copyright{Balagopal Komarath, Anant Kumar, and Akash Pareek} 

{

\ccsdesc[100]{Theory of computation~Graph algorithms analysis}} 

\keywords{Homomorphism Counting, Subgraph Counting, Induced subgraph Counting, Bounded degeneracy graph} 

\category{} 

\acknowledgements{}

\nolinenumbers 

\ArticleNo{1}

\usepackage{algorithm,algpseudocode}
\usepackage{tikz}
\usepackage{enumitem}

\newcommand{\spasm}{\mathsf{Spasm}}
\newcommand{\merge}{\mathsf{Merge}}

\newcommand{\td}{\mathsf{td}}
\newcommand{\gmain}{H}
\newcommand{\SEEN}{\textsc{Seen}}

\begin{document}

\maketitle

\begin{abstract}

For a fixed pattern graph, we study the algorithmic complexity of counting homomorphisms, subgraph isomorphisms, and induced subgraph isomorphisms into an $n$-vertex, $d$-degenerate host graph. Bressan (Algorithmica, 2021) introduced the notion of DAG treewidth and showed that counting homomorphisms and induced subgraphs can be performed efficiently using dynamic programming that requires polynomial space. In this work, we introduce a new graph parameter, called \emph{DAG treedepth}, which enables efficient divide-and-conquer algorithms for counting homomorphisms in $d$-degenerate host graphs using only constant space.

Bera, Gishboliner, Levanzov, Seshadhri, and Shapira (SODA, 2021) showed that a pattern graph has DAG treewidth one if and only if it contains no induced cycle of length at least six. This induced minor characterization leads to linear-time and linear-space algorithms. Building on this line of work, we derive an induced-minor characterization of graphs with DAG treedepth at most two that uses only constant space.

Recently, Paul-Pena and Seshadhri (ICALP, 2025) proved that all pattern graphs on at most nine vertices can be counted in subquadratic time using polynomial space. We show that every pattern graph on at most nine vertices can be counted as an induced subgraph in $O(n^3)$ time using only constant space. Moreover, we show that patterns on at most eleven vertices can be counted in $O(n^2)$ time using polynomial space.

Finally, we present a constant-space algorithm for counting induced subgraphs that matches the running time of Bressan’s algorithm. We further show that, when polynomial space is allowed, homomorphisms, subgraph isomorphisms, and induced subgraph isomorphisms can be counted faster than Bressan’s algorithm. In addition, we establish several other results related to DAG treewidth and DAG treedepth that may be of independent interest.

\end{abstract}
\newpage
\section{Introduction}\label{sec:intro}

For simple, undirected graphs $G$ (host) and $H$ (pattern), we consider the problem of counting the number of occurrences of $H$ in $G$ under the following three fundamental notions of containment.
\begin{enumerate}
    \item The number of subgraphs of $G$ isomorphic to $H$. That is, the cardinality of $\{(V', E') \mid (V', E') \cong H, V' \subseteq V(G), E' \subseteq E(G) \}$.
    \item The number of induced subgraphs of $G$ isomorphic to $H$. That is, the number of vertex subsets $S \subseteq V(G)$ such that the induced subgraph $G[S] = (S, \{\{u,v\} \in E(G) \mid u,v \in S\})$ is isomorphic to $H$.
    \item The number of homomorphisms from $H$ to $G$. That is, the number of functions $\phi : V(H) \to V(G)$ such that for every edge $\{u,v\} \in E(H)$, we have $\{\phi(u),\phi(v)\} \in E(G)$.

\end{enumerate}

Counting subgraphs and homomorphisms is at the heart of many problems in combinatorics, database theory, and network analysis \cite{flum2004parameterized, lovasz2012large, ahmed2015efficient, curticapean2017homomorphisms, deep2025ranked, kara2025conjunctive}. These problems capture fundamental questions such as detecting specific patterns in large networks and evaluating conjunctive queries. When the host graph $G$ is unrestricted, the complexity of such counting tasks is well understood, as it is closely tied to the structural measures of the pattern graph $H$. In particular, the running time for counting subgraphs is governed by the \emph{treewidth} of $H$ in the setting of unrestricted space, and by the treedepth of $H$ when constant space is imposed \cite{DBLP:journals/tcs/DiazST02, komarath2023finding}.

In many real-world settings, host graphs are not arbitrary but exhibit additional structure such as sparsity or bounded degeneracy \cite{pashanasangi2020efficiently, jha2015path}. The class of \emph{$d$-degenerate graphs} provides a natural way to model such sparse networks. A graph is said to be $d$-degenerate if every subgraph contains a vertex of degree at most $d$. This family includes many important graph classes, such as planar graphs, graphs excluding a fixed minor, and preferential attachment
graphs \cite{seshadhri2023some}. Studying counting problems on degenerate graphs leads to more refined algorithmic bounds that better reflect the structure of real data.

Existing work on counting patterns in bounded-degeneracy graphs has primarily focused on improving \emph{time} complexity. By exploiting acyclic orientations of the host graph with bounded outdegree, Bressan~\cite{bressan2021faster} introduced the notion of \emph{DAG treewidth} and obtained faster algorithms for counting homomorphisms and induced subgraphs via dynamic programming. This approach has led to a rich theory of (near) linear-time algorithms and structural characterizations for restricted pattern classes \cite{bera2021near, bressan2022exact, bera2022counting, paul2022dichotomy, paul2025dichotomy}.

In contrast, the \emph{space complexity} of pattern counting in bounded-degeneracy graphs has received little attention. Dynamic programming techniques underlying existing algorithms inherently require maintaining tables whose size grows with the width of the decomposition, and therefore do not extend naturally to the constant-space setting. This motivates the search for structural parameters that support divide-and-conquer algorithms with bounded recursion depth.

In this work, we introduce a new graph parameter, called \emph{DAG treedepth} (dtd), which can be viewed as a treedepth analogue of DAG treewidth. The DAG treedepth parameter captures the reachability structure induced by an acyclic orientation of the pattern graphs. We show that this parameter precisely governs the complexity of counting homomorphisms, subgraphs, and induced subgraphs in bounded-degeneracy graphs.  The algorithms use only constant space, assuming that both the degeneracy and the dtd are constants.

\subsection{Our Results}

We begin by introducing a new structural parameter, called the \emph{DAG treedepth} (see \Cref{def:dtd}), which serves as a natural analogue of treedepth. The DAG treedepth parameter enables the design of efficient, constant-space, divide-and-conquer algorithms for counting homomorphisms into $d$-degenerate host graphs. Our first main result shows that this parameter precisely governs the complexity of counting homomorphisms using only constant space.

\begin{restatable}{theorem}{dtdalgo}
\label{thm:dtdalgo}
Let $d$ be a fixed constant. If the DAG treedepth of a $k$-vertex graph $H$ is $t$, then the number of homomorphisms from $H$ to an $n$-vertex, $d$-degenerate host graph $G$ can be counted in $O(n^{t})$ time using constant space.
\end{restatable}

We further show that the DAG-treedepth framework extends beyond homomorphism counting. In particular, it yields constant-space algorithms for counting induced subgraphs in bounded-degeneracy graphs whose running times asymptotically match the best known bounds obtained via DAG treewidth ~\cite{bressan2021faster}. Crucially, our algorithms avoid the polynomial-space overhead inherent in dynamic programming.

\begin{restatable}{theorem}{indsubcount}
\label{thm:indsubcount}
Let $d$ be a fixed constant and let $H$ be a $k$-vertex graph. Then the number of induced subgraphs of an $n$-vertex, $d$-degenerate graph $G$ that are isomorphic to $H$ can be counted in $O(n^{k/4+O(1)})$ time using constant space.
\end{restatable}

We also show that our techniques yield faster constant-space algorithms for counting subgraphs when the pattern graph is sufficiently sparse, with running time depending on the number of edges rather than the number of vertices. Additionally, in \Cref{thm:indsubcount}, we show that, instead of constant space, if we allow polynomial space, we can improve the bound from $O(n^{k/4+O(1)})$ to $O(n^{k/5+O(1)})$.

\begin{restatable}{theorem}{k-vertex-ind}
    \label{thm:k-vertex-ind}
    Let $H$ be a $k$-vertex graph and let $d$ be a constant. Then there is an $O(n^{k/5+O(1)})$-time and polynomial space algorithm to count the number of induced subgraphs of a given $n$-vertex, $d$-degenerate graph $G$ that are isomorphic to $H$.
\end{restatable}

While DAG treewidth has proved effective for obtaining linear-time algorithms on bounded-degeneracy graphs, its structural complexity quickly becomes intractable beyond this regime. In particular, as we show, even deciding whether a directed acyclic graph has DAG treewidth at most two is \textup{NP}-complete. Therefore, DAG treewidth does not admit induced minor characterizations even in the quadratic-time regime. In contrast, DAG treedepth admits a significantly more tractable structural theory. In particular, for every fixed constant $k$, it can be decided in polynomial time whether a given directed acyclic graph has DAG treedepth at most $k$ (see \Cref{thm:dtdverified}). 

Bera, Pashanasangi, and Seshadhri~\cite{bera2020linear} showed that a pattern graph $H$ has DAG treewidth one if and only if it has no induced cycle of length $6$ or greater. Inspired by the induced-minor characterization of patterns admitting linear-time and linear-space algorithms \cite{bera2020linear}, we develop an analogous framework for DAG treedepth. This framework enables both constant-space algorithms and structural characterizations of tractable pattern classes. In particular, we obtain a complete induced-minor characterization of patterns that admit cubic-time, constant-space algorithms.

\begin{restatable}{theorem}{dagtdtwo}\label{thm:dagtdtwo}
A pattern graph $H$ has DAG treedepth at most two if and only if it is $\{C_6, P_7, H_1, H_2\}$-induced-minor-free (See ~\Cref{fig:h1} and ~\Cref{fig:h2}).
\end{restatable}

More generally, we investigate the structural relationship between induced-minor–based characterizations of bounded DAG treedepth and the classical minor-based characterizations of bounded treedepth (see \Cref{lem:indmin1} and \Cref{lem:indmin2}). We show that, for every fixed integer \(k\), the class of graphs of DAG treedepth less than \(k\) admits an induced-minor characterization that is directly related to the minor obstruction set for treedepth and their supergraphs.

We further complement this structural connection with conditional lower bounds, relating the complexity of algorithms parameterized by DAG treedepth to that of algorithms parameterized by treedepth. In particular, improving the running time of constant-space algorithms for patterns of DAG treedepth three would imply a corresponding breakthrough for classical pattern-counting problems.

\begin{restatable}{theorem}{thmlower}\label{thm:dtd3-lower}
For every $\varepsilon > 0$, there is no $O(n^{3-\varepsilon})$-time constant-space algorithm for counting subgraphs isomorphic to any pattern $H$ with $dtd(H)=3$ in bounded-degeneracy graphs of degeneracy two, unless there exists an $O(n^{3-\varepsilon})$-time constant-space algorithm for counting triangles in arbitrary graphs.
\end{restatable}
Note that subcubic-time algorithms for triangle counting based on fast matrix multiplication are known. However, such algorithms inherently require polynomial space and therefore fall outside the constant-space setting considered here.

Several graph patterns can be counted efficiently when the host graph has bounded degeneracy. Bressan established efficient counting results for certain special classes of patterns, including \(K_n\), \(K_n - \{e\}\), and $K_{k_1,k_2}$. In~\cite{bera2020linear}, the authors showed that all patterns on at most five vertices can be counted in nearly linear time using linear space. More recently, Paul-Pena and Seshadhri \cite{paul2024subgraph} proved that all patterns on at most nine vertices can be counted in subquadratic time using polynomial space. To this end, we show the following:
\begin{theorem}
Let $d$ be a constant. For all patterns with at most nine vertices, we can count the number of occurrences as induced subgraphs in $O(n^3)$ time and constant space for an $n$-vertex $d$-degenerate graph given as input.
\end{theorem}
We further show that the bound of nine vertices in the above theorem is tight by establishing a conditional lower bound for a specific pattern on ten vertices.

 Paul-Pena and Seshadhri~\cite{paul2024subgraph} also identified a specific ten-vertex pattern, denoted $\mathcal{H}_{\Delta}$, and conjectured that it does not admit a sub-quadratic time polynomial space algorithm.

\begin{restatable}{conjecture}{conjjj}
\label{con:paul}
(Conjecture~1.8 of~\cite{paul2024subgraph})
For any $\varepsilon > 0$, there is no $O(m^{2-\varepsilon})$-time
\emph{combinatorial} algorithm for computing $\mathrm{sub}(\mathcal{H}_{\Delta}, G)$.
\end{restatable}
We resolve this conjecture affirmatively in the combinatorial setting by assuming that counting $K_4$-copies needs (for combinatorial algorithms) quadratic time in the number of edges.

Beyond resolving this conjecture, we provide a structural explanation for the observed nine-vertex barrier. We show that every pattern graph on at most eleven vertices admits an acyclic orientation whose DAG treewidth is at most two. We also show that this bound is tight by providing a graph on twelve vertices whose DAG treewidth is three. 

\begin{theorem}
Let $d$ be a constant. For every pattern graph $H$ with at most eleven vertices, the number of induced subgraphs of an $n$-vertex $d$-degenerate graph $G$ that are isomorphic to $H$ can be counted in $O(n^2)$ time using polynomial space. 
\end{theorem}

We further establish a relation between DAG treedepth and treedepth of the pattern graph, as well as DAG treewidth of the pattern graph. In addition, we show many properties and bounds related to DAG treewidth and DAG treedepth, which may be of independent interest.

\subsection{Related work}

Subgraph counting is a fundamental problem in graph algorithms, with its complexity and tractability heavily influenced by the structural properties of the host graph. In \emph{nowhere dense} and \emph{bounded degeneracy} graphs, several works have established exact and approximate algorithms, along with complexity dichotomies. A recent dichotomy hierarchy precisely characterizes linear-time subgraph counting in such graphs~\cite{paul2025dichotomy}. In the dense graph regime, Bressan et al. \cite{bera2021near} classified pattern counting complexity using DAG-treewidth, providing a sharp dichotomy for somewhere dense graph classes~\cite{bressan2023somewheredense}. At the same time, related work has addressed the complexity of counting homomorphic cycles in degenerate graphs~\cite{gishboliner2023counting}, and pattern counting in directed graphs parameterized by outdegree~\cite{bressan2023complexity, bressan2022exact}. Recent studies have also considered counting patterns when the host graph is sparse~\cite{komarath2023finding}. For approximate counting, several works have proposed efficient approximation algorithms for subgraph counting problems, including a general approximation framework for sparse graphs~\cite{lokshtanov2025efficiently}, and fast approximation algorithms specifically for triangle counting in streaming and distributed settings~\cite{bera2020degeneracy}. Beyond graphs, subhypergraph counting has recently gained traction. In ~\cite{bressan2025complexity}, the authors extended the homomorphism basis framework to hypergraphs, while~\cite{paul2025near} characterized the class of pattern hypergraphs $H$ that can be counted in near-linear time when the host graph $G$ is a hypergraph.

\subsection{Paper Organization}
We begin with the necessary preliminaries, introducing the definitions, notation, and background used throughout the paper. In \Cref{sec:treedepth}, we define DAG treedepth and prove \Cref{thm:dtdalgo}. We then identify the obstructions for DAG treedepth $0$, $1$, and $2$ in \Cref{sec:dtdobstraction}. Subsequently, in \Cref{sec:dtdlower}, we establish conditional lower bounds for counting patterns in cubic time using constant space. In \Cref{sec:dtdconstant}, we present constant-space algorithms for computing $\mathrm{hom}(H, G)$, $\mathrm{sub}(H, G)$, and $\mathrm{ind}(H, G)$.

 We then shift our focus to DAG treewidth in \Cref{sec:treewidth}. In \Cref{sec:dtwproperty}, we present key properties of DAG treewidth, which are subsequently used in \Cref{sec:dtwfast} to design faster algorithms for computing $\mathrm{hom}(H, G)$, $\mathrm{sub}(H, G)$, and $\mathrm{ind}(H, G)$, along with several additional results. Finally, in \Cref{sec:connection}, we establish relationships between DAG treedepth and treedepth, as well as between DAG treewidth and treewidth.
\section{Preliminaries}\label{sec:prelim}

In this section, we present several definitions that will be used throughout the paper. All graphs considered are simple and may be either directed or undirected. For standard terminology in graph theory, we refer the reader to the textbook by Douglas West~\cite{west2001introduction}. We use the following standard notations for some well-known graphs: $e$ or ($uv$) for edges in a graph, $P_k$ for $k$-vertex paths, $C_k$ for $k$-cycles, $K_k$ for $k$-cliques, $K_k - e$ for a $k$-clique with one edge missing, $K_{m,n}$ for complete bipartite graphs. A $k$-star is a $(k+1)$-vertex graph with a vertex $u$ adjacent to vertices $v_1, \dotsc, v_k$ and no other edges. A \emph{star graph} is a $k$-star for some $k$. For a graph $G$ and $S \subseteq V(G)$, we denote by $G[S]$ the subgraph of $G$ induced by the vertices in $S$.
We begin with the definition of graph homomorphisms.
\begin{definition}
\label{def:hom}(Graph Homomorphism)
    Given two graphs $G$ and $H$, a graph homomorphism from $G$ to $H$ is a map $\phi: V(G) \rightarrow V(H)$ such that if $(u,v) \in E(G)$, then $(\phi(u),\phi(v)) \in E(H)$.
    
    A homomorphism $\phi$ is called a \emph{partial homomorphism} if $\phi(v)$ is not necessarily defined for all $v \in V(G)$ and if $(u, v) \in E(G)$ and $\phi(u)$ and $\phi(v)$ are defined, then $(\phi(u), \phi(v)) \in E(H)$.
\end{definition}

We now recall several fundamental graph-theoretic parameters, including treewidth, treedepth, and a few others.

\begin{definition}(Tree Decomposition)\label{def:treedecomposition}
    A \emph{tree decomposition} of a graph $G$ is a pair $(T, \{B(t)\}_{t \in V(T)})$, where $T$ is a tree and each node $t \in V(T)$ is assigned a subset $B(t) \subseteq V(G)$, called a \emph{bag}.
\end{definition}

A tree decomposition of a graph $G$ has the following properties:

\begin{itemize}
    \item \textbf{Connectivity Property:} For every vertex $v \in V(G)$, the set of nodes $\{t \in V(T) \mid v \in B(t)\}$ forms a connected component in $T$.

    \item \textbf{Edge Property:} For every edge $e = \{u, v\} \in E(G)$, there exists a node $t \in V(T)$ such that both $u$ and $v$ are in $B(t)$.
\end{itemize}

\begin{definition}(Treewidth)\label{def:treewidth}
    Let the width of a tree decomposition be defined as $\max_{t \in V(T)} |B(t)| - 1.$ Then, the \emph{treewidth} of a graph $G$, denoted as $\mathrm{tw}(G)$, is the minimum width over all possible tree decompositions of $G$.
\end{definition}

\begin{definition}(Elimination Tree)\label{def:eliminationtree}
An \emph{elimination tree} of a connected graph $G$ is a rooted tree $(T, r)$, where $r \in V(G)$ is the root, and each subtree rooted at a child of $r$ is an elimination tree of a connected component of the graph $G \setminus \{r\}$. For an empty graph (i.e., a graph with no vertices), the elimination tree is defined to be empty.

 The \emph{depth} of an elimination tree $(T, r)$ is the length (i.e., the number of vertices) of the longest path from the root $r$ to any leaf in the tree.

\end{definition}

\begin{definition}(Treedepth)\label{def:treedepth}
    The \emph{treedepth} of a graph $G$, denoted as $\td(G)$, is the minimum depth among all possible elimination trees of $G$.
\end{definition}

\begin{definition}(Edge Contraction)\label{def:edgecontract}
    An \emph{edge contraction} is an operation in which an edge $(u, v)$ is removed from the graph and its two vertices $u$ and $v$ are merged into a single vertex $uv$. All edges incident to $u$ or $v$ are now considered incident to the new vertex $uv$.
\end{definition}

\begin{definition}(Minor and Induced Minor)\label{def:minorindminor}
    A graph $G$ is said to be a \emph{minor} of another graph $G'$ if $G$ can be obtained from $G'$ by deleting vertices, deleting edges, and contracting edges. A graph $G$ is an \emph{induced minor} of a graph $G'$ if it can be obtained by contracting edges and deleting vertices of $G'$.
\end{definition}

Let $H$ be a pattern graph and $G$ be the host graph. We denote the number of homomorphism from $H$ to $G$ by $\mathrm{hom(H,G)}$, number of subgraphs of $H$ in $G$ by $\mathrm{sub}(H,G)$, and number of induced subgraph of $H$ in $G$  by $\mathrm{ind}(H,G)$.

\begin{definition}(Spasm)\label{def:spasm}
Let $H$ be a graph. If $H$ is a clique, then $\spasm(H) = \{H\}$.
Otherwise, let $\mathcal{I}$ denote the set of all independent sets of $H$ of size at least two.  
For any $I \in \mathcal{I}$, let $\merge(H, I)$ be the graph obtained by merging all vertices in $I$ into a single vertex. Then, $\spasm(H) = \{H\} \cup \bigcup_{I \in \mathcal{I}} \spasm(\merge(H, I)).$

\end{definition}

Note that for any fixed pattern graph $H$, the number of subgraph isomorphisms from $H$ to a host graph $G$ can be expressed as a linear combination of homomorphism counts from graphs in $\spasm(H)$ to $G$. 
More precisely, there exist constants $\{\alpha_{H'}\}_{H' \in \spasm(H)}$, such that
$\mathrm{sub}(H,G) \;=\; \sum_{H' \in \spasm(H)} \alpha_{H'} \,\mathrm{hom}(H',G)$. This identity is a key tool in the literature for counting subgraphs (see, for example, \cite{Alon1997, 10.1145/3055399.3055502})

\begin{definition}(Directed Acyclic Graph (DAG))
    Given an undirected graph $H = (V_H, E_H)$, an acyclic orientation of $H$ is an assignment of a direction to each edge $\{u, v\} \in E_H$, converting it to either $(u \rightarrow v)$ or $(v \rightarrow u)$, such that the resulting directed graph $\vec{H}$ is a directed acyclic graph (DAG) that is, $\vec{H}$ contains no directed cycles.
\end{definition}

Let $\vec{H}$ be a directed acyclic graph (DAG), and let $S = S(\vec{H})$ denote the set of \emph{sources} in $\vec{H}$, i.e., the set of vertices with in-degree zero. Otherwise, we denote $u \in V(\vec{H}\setminus S)$ as a non-source vertex. For any two vertices $u, v \in V(\vec{H})$, we say that \emph{$v$ is reachable from $u$} if there exists a directed path from $u$ to $v$ in $\vec{H}$.

For a source vertex $s \in S$, let $\mathrm{R}_H(s)$ denote the set of all vertices reachable from $s$ in $\vec{H}$. More generally, for a set of sources $B = \{s_1, s_2, \ldots, s_\ell\} \subseteq S$, we define, $\mathrm{R}_H(B) = \bigcup_{i=1}^{\ell} \mathrm{R}_H(s_i).$ When the underlying graph is clear from the context, we define the set of reachable vertices from $s$ as $\mathrm{R}(s)$.

\begin{definition}(Subdivision Vertex)
Let $G = (V,E)$ be an undirected graph and let $\{u,v\} \in E$.
The graph obtained after subdividing the edge $\{u,v\}$ is the graph
$G' = (V \cup \{w\},\, (E \setminus \{\{u,v\}\}) \cup \{\{u,w\}, \{w,v\}\})$, where $w$ is a new vertex not in $V$.
The vertex $w$ is called a \emph{subdivision vertex}.
\end{definition}

\begin{definition}(Hypergraph)
    A hypergraph is a pair $\mathcal{H}=(V, E)$, where $V$ is a finite set of vertices, and $E \subseteq 2^{V}$ is a non-empty subset of $V$, called hyperedges.
\end{definition}

\begin{definition}(Generalized hypertree decomposition (GHD)) A GHD of a hypergraph $\mathcal{H}=(V,E)$ is a tuple $\langle T, (B_u)_{u \in N(T)}, (\lambda_u)_{u \in N(T)} \rangle$, such that $T=\langle N(T),E(T)\rangle$ is a rooted tree and the following conditions holds:

\begin{enumerate}
    \item for each $e \in E$, there is a node $u \in N(T)$ with $e \subseteq B_u$.

    \item for each $v \in V$, the set $\{u \in N(T) | v \in B_u\}$ is connected in $T$.

    \item for each $u \in N(T)$, $\lambda_u$ is a function $\lambda_u: E \rightarrow |\{0,1\}$ with $B_u \subseteq B(\lambda_u)$
\end{enumerate}
    
\end{definition}

\begin{definition}(Generalized hypertree width of $\mathcal{H}$ ($ghw(\mathcal{H})$))
The width of a GHD is the maximum weight of the function $\lambda_u$ over all nodes $u$ in $T$. The $ghw(\mathcal{H})$ is defined as the minimum width of all GHDs.
\end{definition}

 In \cite{bressan2021faster}, Bressan described a way to construct a bipartite graph $\textsc{Bip}(\vec{H})$ from any DAG $\vec{H}$. For completeness, we briefly explain this construction and then describe a simpler version called $G_S$, which only uses the source vertices of $\vec{H}$. We will use $G_S$ later in several results to get useful bounds.

\subsection{Construction of $G_S$ from $\textsc{Bip}(\vec{H})$}\label{sec:construction}
Given a directed acyclic graph $\vec{H}$ with a fixed acyclic orientation, Bressan \cite{bressan2021faster} constructed a bipartite graph $\textsc{Bip}(\vec{H})$ as follows:
\begin{itemize}
    \item The vertex set of $\textsc{Bip}(\vec{H})$ denoted as $V(\textsc{Bip}(\vec{H})) = (S, V(H) \setminus S)$, where $S$ is the set of sources in $\vec{H}$.

    \item An edge $\{u, v\} \in E(\textsc{Bip}(\vec{H}))$ if and only if $v \in \mathrm{R}_H(u)$, i.e., $v$ is reachable from the source vertex $u$ in $\vec{H}$.

\end{itemize}

Here, every edge in $\textsc{Bip}(\vec{H})$ has one of its endpoints as a source vertex and another endpoint as a non-source vertex. Now, we perform edge contraction. Note that one non-source vertex may be reachable from several sources and therefore have several options for edge contraction. We sequentially contract exactly one edge for each non-source and get the resultant graph. Note that the vertices of the resultant graph are exactly the set of source vertices. The resulting graph may differ depending on the choice of edge contraction. We denote by $\mathcal{G}_{\vec{H}}$ (we drop the subscript when the DAG is clear from the context), the set of all such graphs. We will later demonstrate how to utilize this family of graphs to derive several bounds for DAG treedepth and DAG treewidth.

\subsection{Model of Computation} The model of computation that we consider is the unit-cost RAM model. In particular, we can store labels of vertices and edges of the host graph in a constant number of words\footnote{In the TM model or the log-cost RAM model, storing labels of vertices would take $O(\log n)$ space.}. In this model, algorithms based on fast matrix multiplication and/or treewidth mentioned above use polynomial space. However, the brute-force search algorithm uses only constant space as it only needs to store $k$ vertex labels at a time (Recall that we regard $k$ as a constant).

\section{DAG Treedepth}\label{sec:treedepth}

In \cite{bressan2021faster}, Bressan introduced DAG treewidth ($\mathrm{dtw}$), a structural parameter analogous to treewidth, and used it to design efficient algorithms for computing $\mathrm{hom}(H, G)$, $\mathrm{sub}(H, G)$, and $\mathrm{ind}(H, G)$. Motivated by this, we define a related parameter called DAG treedepth ($\mathrm{dtd}$), analogous to treedepth. While DAG treewidth enables the design of fast algorithms but suffers from hardness constraints, DAG treedepth allows the design of constant-space algorithms and does not exhibit the same intractability barriers.

In this section, we first formally define DAG treedepth and prove \Cref{thm:dtdalgo}. In ~\Cref{sec:dtdobstraction}, we identify the obstruction sets for DAG treedepth 0, 1, and 2. In \Cref{sec:dtdlower} we present conditional lower bounds for pattern counting and a few results on small pattern counting. 
In \Cref{sec:dtdconstant}, we compute $\mathrm{hom}(H, G)$, $\mathrm{sub}(H, G)$, and $\mathrm{ind}(H, G)$ using constant space, improving the space complexity of the fast algorithms of \cite{bressan2021faster} from polynomial to constant. We also provide upper bounds using the number of vertices and the number of edges of the pattern graph $H$.

\begin{definition}
    \label{def:dtd}
    Let $\vec{H}$ be any DAG. A \emph{DAG elimination forest} of $\vec{H}$ is defined as a collection of rooted trees, constructed recursively as follows:
    \begin{itemize}
        \item If $\vec{H}$ is empty, the forest is empty.
        \item If the underlying undirected graph of $\vec{H}$ is disconnected, then construct DAG elimination trees for each connected component, and their union is a DAG elimination forest for $\vec{H}$.
        \item Otherwise, $\vec{H}$ has at least one source and its underlying undirected graph is connected. In this case, the forest is a tree. Pick a source $s$ arbitrarily, delete that source and all vertices reachable from it in $\vec{H}$. The root of the tree is $s$, and its sub-trees are the trees in the DAG elimination forest for the remaining DAG.
    \end{itemize}

    The \emph{DAG treedepth} of a DAG $\vec{H}$ is defined as the minimum depth of any elimination forest constructed in the above fashion. For an undirected graph $H$, its DAG treedepth is defined as the maximum DAG treedepth taken over all acyclic orientations of $\vec{H}$.

\end{definition}

\begin{remark}
Notice that the definition of DAG treedepth is analogous to the definition of treedepth except that at each step we can delete a source and all non-sources reachable from it, instead of only a single node. For DAG treedepth, the nodes of the DAG elimination forest correspond exactly to the DAG's sources.
\end{remark}

\begin{observation}
  ~\Cref{fig:dadexample} shows two examples where the treedepth is unbounded. For instance, cliques and complete bipartite graphs have unbounded treedepth. However, the DAG treedepth of a clique is 1, and for complete bipartite graphs, it is 2.
\end{observation}

\begin{figure}[h]
    \centering

    \begin{subfigure}[b]{0.48\textwidth}
        \centering
        \includegraphics[width=\textwidth]{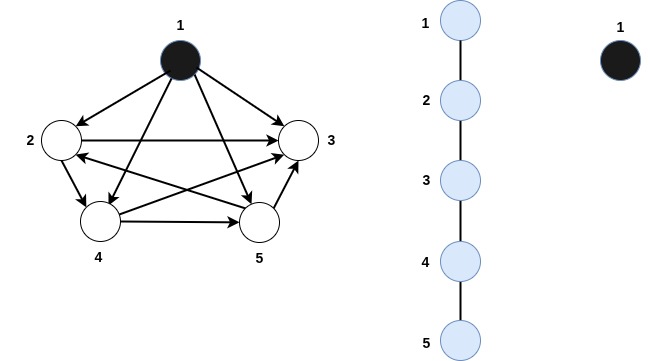}
        \caption{The first figure is $K_5$ with one source. The second figure is the elimination tree of $K_5$. The third figure is the DAG Elimination tree of $K_5$.}
        \label{fig:dtdclique}
    \end{subfigure}
    \hfill
    \begin{subfigure}[b]{0.48\textwidth}
        \centering
        \includegraphics[width=\textwidth]{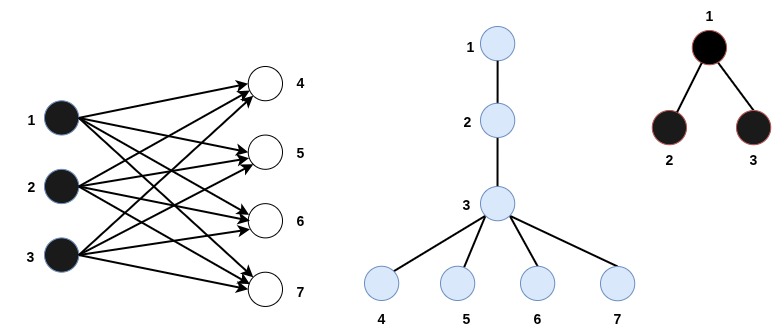}
        \caption{The first figure is $K_{3,4}$ with 3 sources. The second figure is the elimination tree of $K_{3,4}$. The third figure is the DAG Elimination tree of $K_{3,4}$.}
        \label{fig:dtdkmn}
    \end{subfigure}

    \caption{The black nodes are sources, and the white nodes are non-sources. In figure (a), the DAG treedepth is 1, but the treedepth is $5$. In figure (b), DAG treedepth is 2, but treedepth is 4.}
    \label{fig:dadexample}
\end{figure}

\begin{algorithm}[h]
\caption{COUNT-HOM(G, T, v, H, $\sigma$)\label{alg:count-hom-mtd}}
\begin{algorithmic}[1]
\Require{$G$ - The $d$-degenerate host graph}
\Require{$T$ - The DAG elimination tree for $H$}
\Require{$v$ - A vertex in $T$ (or $H$)}
\Require{$H$ - The pattern DAG}
\Require{$\sigma$ - A partial homomorphism from $H$ to $G$}
\State $c \gets 0$
\ForAll{$u\in V(G)$}
  \State $p \gets 0$
  \ForAll{$\sigma'$ extending $\sigma$ to $\{v\} \cup R(v)$ such that $\sigma'(v) = u$}
    \State $p \gets 1$
    \ForAll{children $w$ of $v$ in $T$}
      \State $p \gets p \times \text{COUNT-HOM(G, T, w, H, $\sigma'$)}$
    \EndFor
  \EndFor
  \State $c \gets c + p$
\EndFor
\State \Return $c$
\end{algorithmic}
\end{algorithm}

Next, we provide one of our main theorems, which is implied by \Cref{alg:count-hom-mtd}.
\dtdalgo*

\begin{proof}
    See \Cref{alg:count-hom-mtd}. The required count is obtained by \texttt{COUNT-HOM(G, T, r, H, $\{\}$}), where $T$ is the DAG elimination tree of depth $t$ and $r$ is the root of $T$.
    
    We claim that \texttt{COUNT-HOM(G, T, v, H, $\sigma$)} correctly computes the number of homomorphisms that extend the partial homomorphism $\sigma$ to all the sources in the sub-tree of $T$ rooted at $v$ and all the non-sources reachable from those sources when $\sigma$ is a partial homomorphism that maps all ancestors of $v$ in $T$ and all non-sources reachable from those vertices. We prove this using an induction on the height of the node $v$ in $T$. 

    In the base case, $v$ is a leaf. In this case, the algorithm is correct because we are simply iterating over all possible mappings $v$. Any non-source reachable from $v$ is already mapped in $\sigma$ or is only reachable from $v$. So once we fix the image of $v$ and all vertices only reachable from it, we only have to check whether these added images to $\sigma$ are consistent with the other defined vertices in $\sigma$. This is constant-time.

    If $v$ is an internal node, then observe that after mapping $v$ and all non-sources reachable from it, the subgraphs of $H$ spanned by subtrees of $v$ in $T$ are disjoint. That is, if $u$ and $w$ are two sources in separate subtrees of $v$ in $T$, then for any non-source $x$ reachable from $v$, $w$ is also reachable from $v$ or one of its ancestors. So $\sigma(x)$ is defined. We can count extensions of $\sigma$ to sources in each of the subtrees and new non-sources reachable from them independently and compute the total by multiplying these counts.

    In each iteration of the outer loop of \Cref{alg:count-hom-mtd}, observe that the loop in line~4 can only have $g(k, d)$ many iterations for some function $g$ as there are at most $d$ outgoing edges from any vertex in $H$. The length of any simple path from $v$ in $H$ is bounded by $k$. The number of iterations of the loop in line~6 is bounded by $k$. When a recursive call is made, the depth of $v$ increases by $1$. Therefore, the time complexity is $f(k, d)n^t = O(n^t)$ for some function $f$. Recall that we consider $k$ and $d$ as constants. So we can absorb them into the $O(.)$ notation. Each level of the recursion uses only constant space, and the depth of the recursion is at most $t$. Since we regard $t$ as a constant, the algorithm uses only constant space.
\end{proof}

Having established that bounded DAG treedepth enables constant-time homomorphism counting, a natural algorithmic question arises: given a directed graph $G$, how can one determine its DAG treedepth? In particular, since our main result assumes the parameter as part of the input, it is essential to understand the computational complexity of detecting and computing this parameter. We therefore study the problem of deciding whether a graph has DAG treedepth at most $k$, and of constructing a corresponding decomposition when it exists.

\begin{theorem}\label{thm:dtdverified}
 Given a DAG $H$, it can be verified in $O(g(k)n^{f(k)})$ time whether $dtd(H)\le k$ or not, where $f(k)$ and $g(k)$ are some computable functions.
\end{theorem}
\begin{proof}
    We can design an XP algorithm to check whether, for a given DAG $H$, $\operatorname{dtd}(H) \le k$.  
The algorithm (\Cref{alg:dtd-check}) recursively constructs a DAG elimination tree that satisfies the required properties.

We proceed as follows.  
Initially, we pick a source vertex $s$ of $H$, delete $s$ and all non-source vertices reachable from it (denoted by $R(s)$), and obtain the reduced DAG $H' = H - \{s\} - R(s)$.  
In the elimination tree, $s$ becomes the root.

If $H'$ is connected, we recursively pick a source $s_1$ from $H'$, delete $s_1$ and its reachable vertices, and make $s_1$ a child of $s$ in the elimination tree.

If $H'$ is disconnected, say $H'$ has connected components $H_1', H_2', \dots, H_r'$, we recursively apply the same process to each $H_i'$.  
In the elimination tree, the roots of the elimination trees of $H_i'$ are made children of $s$.

The above process continues until the graph becomes empty.  
If there exists a sequence of choices of sources that yields an elimination tree of height at most $k$, then $\operatorname{dtd}(H) \le k$.

Since at each step we try all possible choices of sources and the depth of recursion is at most $k$, and there can be at most $O(n)$ sources at any step, the running time is bounded by $O(g(k)n^{f(k)})$.  
Hence, the algorithm runs in XP time with respect to parameter $k$.

\begin{algorithm}[h]
\caption{Checking whether $\operatorname{dtd}(H) \le k$ for a DAG $H$}
\label{alg:dtd-check}
\begin{algorithmic}[1]
\Require{A directed acyclic graph $H = (V,E)$ and an integer $k \ge 1$}
\Ensure{Returns \textbf{true} if $\operatorname{dtd}(H) \le k$, otherwise \textbf{false}}
\Function{CheckDTD}{$H, k$}
    \If{$H$ is empty}
        \State \Return \textbf{true}
    \EndIf
    \If{$k < 1$}
        \State \Return \textbf{false}
    \EndIf
    \For{each source vertex $s$ of $H$}
        \State Let $R(s)$ be the set of non-source vertices reachable from $s$
        \State Let $H' = H - \{s\} - R(s)$
        \If{$H'$ is disconnected}
            \State Let $\{H_1', H_2', \ldots, H_r'\}$ be the connected components of $H'$
            \State $\textit{valid} \gets \textbf{true}$
            \For{each component $H_i'$}
                \If{\textsc{CheckDTD}$(H_i', k-1)$ is \textbf{false}}
                    \State $\textit{valid} \gets \textbf{false}$
                    \State \textbf{break}
                \EndIf
            \EndFor
            \If{$\textit{valid}$}
                \State \Return \textbf{true}
            \EndIf
        \Else
            \If{\textsc{CheckDTD}$(H', k-1)$ is \textbf{true}}
                \State \Return \textbf{true}
            \EndIf
        \EndIf
    \EndFor
    \State \Return \textbf{false}
\EndFunction
\end{algorithmic}
\end{algorithm}

\end{proof}

Next, we analyze the structure of the DAG elimination tree and derive a structural constraint on the sources. Intuitively, if two sources have a unique common reachable vertex, then any elimination tree must reflect this dependency by placing them along a common root-to-leaf path. The following lemma formalizes this observation.

\begin{lemma}\label{lem:uniqueintersection}
Let $\vec{H}$ be a directed acyclic graph (DAG) with a fixed orientation. Suppose there are two source vertices $u$ and $v$ such that there exists a vertex $w \in \mathrm{R}_H(u) \cap \mathrm{R}_H(v)$, and $w \notin \mathrm{R}_H(s)$ for any other source $s \neq u, v$. Then, in any DAG elimination tree of $\vec{H}$, the sources $u$ and $v$ must lie on the same root-to-leaf path.
\end{lemma}
\begin{proof}
    Assume, for contradiction, that there exists a DAG elimination tree in which $u$ and $v$ do not appear on the same root-to-leaf path. Then, by the definition of an elimination tree, $u$ and $v$ belong to different subtrees, and the sets of vertices discovered by eliminating $u$ and $v$ are disjoint. However, since $w$ is reachable from both $u$ and $v$ and from no other source, the elimination tree would fail to account for the reachability of $w$, violating the correctness of the elimination tree. Therefore, $u$ and $v$ must lie on the same root-to-leaf path.
\end{proof}

\begin{remark}
 Note that, unfortunately, $dtd$ is not minor closed like $td$. One can observe that $K_6$ has only one source and, therefore, $dtd(K_6)=1$. However, considering the $C_6$ subgraph of $K_6$, we know that $dtd(C_6)=3$. Thus, $dtd$ is not even subgraph-closed.
\end{remark}

Let $I$ be an induced minor of a graph $\gmain$, then we show the following theorem.

\begin{theorem}\label{dtdinducedminor}
If $I$ is an induced minor of $\gmain$, then $dtd(\gmain)\geq dtd(I)$.
\end{theorem}
\begin{proof}
    To obtain this result, it suffices to prove the following two claims.

    \begin{enumerate}
        \item If $I$ is an induced subgraph of $\gmain$ then $dtd(\gmain)\geq dtd(I)$.
        \item Let ${H}$ be an undirected graph with $\mathrm{dtd}({H}) = k$. Let $H'$ be the undirected graph obtained by contracting a single edge in $H$. Then, $\mathrm{dtd}({H'}) \leq k.$
    \end{enumerate}

\paragraph*{Proof of Claim 1:}
Let $dtd(I) = k$. By definition, there exists an acyclic orientation $O_1$ of $I$ such that $dtd(I) = k$ under this orientation.

We now construct an acyclic orientation of $\gmain$ that extends $O_1$ and preserves the treedepth. Let $I' \subseteq \gmain$ be an induced subgraph isomorphic to $I$, with an isomorphism $f: V(I) \to V(I')$. Fix the orientation of $I'$ in $\gmain$ according to $O_1$.

Let $V(\gmain) = V(I') \cup \left(V(\gmain) \setminus V(I')\right)$. Partition the remaining vertices based on their distance from $I'$ in $\gmain$: define layers $D_1, D_2, \dots, D_j$, where

$$D_i = \{ v \in V(\gmain) \setminus V(I') \mid \text{dist}_{\gmain}(v, V(I')) = i \}.$$

Now, define a topological ordering of the vertices of $\gmain$ as:

$$V(I'), D_1, D_2, \dots, D_j.$$

Orient the edges from each $D_i$ to $D_{i+1}$ in the forward direction (i.e., from smaller to larger layers), and similarly orient the edges from $V(I')$ to $D_1$ forward. For edges within any $D_i$, choose an arbitrary acyclic orientation.

Note that after this orientation, we have $S(I)=S(\vec{H})\cap I'$. Also, for any non-source $u$ in $I$, we have $P_{I}(u)=P_{\vec{H}}(u)$ i.e.,  there is no source reaching $u$ other than the source vertices in $I$.

Now, suppose for contradiction that $dtd(\gmain) < k$ under this orientation. Let $T$ be a valid DAG elimination tree of $\gmain$ with DAG treedepth less than $k$. Then, the restriction of $T$ to the subgraph $I'$ forms a valid DAG elimination tree of $I'$, and hence of $I$ (by isomorphism), contradicting the assumption that $dtd(I) = k$. Therefore, $dtd(\gmain) \geq dtd(I)$.

\paragraph*{Proof of Claim 2:}

The DAG treedepth of an undirected graph $H$ is the maximum of the DAG treedepth of DAG $\vec{H}$. Let $H'$ be the graph obtained after contracting an edge $\{u,v\}$. We want to show that $dtd(H')\leq dtd(H)$. We pick an arbitrary but fixed acyclic orientation of $H'$. We copy the same orientation in $H$. Note that $\{u,v\}\notin E(H')$. So, the orientation is not known. We first give $(u,v)$ orientation to $\{u,v\}$. Note that the orientation of $H$ is acyclic. Let $w$ be the vertex in $H'$ obtained after contracting $\{u,v\}$. Let $T$ be a DAG elimination tree of $\vec{H}$ of width at most $k$. We will now construct the DAG elimination tree $T'$ of $\vec{H'}$.

\begin{itemize}
    \item \textbf{Case 1:} Both $u$ and $v$ are non-source vertices.

    In this case, we can assume without loss of generality that $P_H(u) \subseteq P_H(v)$, i.e., all source vertices reaching $u$ also reach $v$. Note that the set of source vertices remains unchanged after contraction, i.e., $S(\vec{H'}) = S(\vec{H})$. We consider the same DAG elimination tree decomposition $T$ of $\vec{H}$. One can see that $P_H(v) = P_{H'}(w)$. So, the source vertices in $P_{H'}(w)$ follow the ancestor-descendant relation.

\item \textbf{Case 2:} $u$ is a source vertex.

We consider two subcases depending on whether the new vertex $w$ is a source in $\vec{H'}$ or not.

\begin{itemize}
    \item  \textbf{Case A :} $w$ is a source vertex.

This happens only when $P_H(v) = \{u\}$, meaning $v$ was reachable only from $u$. In this case, $R_H(u) = R_{H'}(w) \setminus \{v\}$. We can construct a DAG elimination tree $T'$ of $\vec{H'}$ by simply replacing $u$ with $w$. Since no new vertices are introduced and treedepth remains unchanged, $\mathrm{dtd}(\vec{H'}) \leq k$.

\item \textbf{Case B:} $w$ is a non-source vertex.

Here, the contraction removes $u$ from the set of sources, so $S(\vec{H'}) = S(\vec{H}) \setminus \{u\}$. $v$ must be reachable by some source vertex $s$ other than $u$. Moreover, $$P_{H'}(w) = P_H(v) \setminus \{u\},$$
and for each source $s \in P_{H'}(w)$, $$R_{H'}(s) = R_H(s) \cup R_H(u).$$ 
Since $P_H(v)$ follows ancestor descendant relations in $T$, and since $u \in P_H(v)$, the subtree formed by $P_H(v)$ also includes $u$. If $u$ appears as an ancestor, then we replace $u$ by $s\in P_{H'}(w)$, which is closest to the root. As these modifications only replace one source with another and do not increase the depth of the tree, the resulting DAG elimination tree $T'$ of $\vec{H'}$ has the same width as $T$, i.e., $\mathrm{dtd}(\vec{H'}) \leq k.$

\end{itemize}

\end{itemize}
Finally, after repetitive use of Claim 2, we get an induced subgraph, and then we use Claim 1. Thus, $dtd(I)\leq dtd(H)$.
\end{proof}

\subsection{DAG Treedepth Obstruction}\label{sec:dtdobstraction}
In this section, we obtain a complete induced-minor characterization of patterns with DAG treedepth at most $2$. This helps in identifying exactly those patterns that admit cubic-time, constant-space algorithms. In addition, we show that, for every fixed integer \(k\), the class of graphs of DAG treedepth less than \(k\) admits an induced-minor characterization that is directly related to the minor obstruction set for treedepth and their supergraphs.

\begin{definition}(Reduced DAG Elimination Tree)
    A DAG elimination tree is called a \emph{reduced DAG elimination tree} if for every source vertex $s$, there exists a non-source vertex $u$ such that in the elimination tree, $u\in R(\textsc{Parent}(s))\cap R(T(s))$ and $u\notin R(\text{ancestor of \textsc{Parent}(s))}.$
\end{definition}

One can check that the leaf source vertex has a unique intersection with its parent source in the DAG elimination tree $T$. We now establish the following lemma.

\begin{lemma}\label{lem:eliminatio}
    Let $H$ be a DAG with $dtd(H) = d$. Then, there exists a reduced DAG elimination tree of $H$ with maximum depth $d$.
\end{lemma}
\begin{proof}
    Suppose that in a DAG elimination tree $T$, for every source vertex $s$, there does not exist a non-source vertex $u$ such that $u \in R(\textsc{Parent}(s)) \cap R(T(s)) \quad \text{and} \quad u \notin R(\text{ancestor of } \textsc{Parent}(s)).$ In this case, we can modify the tree by removing the edge between $s$ and $\textsc{Parent}(s)$, and instead adding an edge between $s$ and $\textsc{Parent}(\textsc{Parent}(s))$.

    We apply the above modification in a bottom-up fashion. That is, we start with the leaf source vertices and ensure each one is non-conflicting. Then we proceed upward through the tree, level by level. At each step, this process reduces the number of conflicts in the elimination tree. Moreover, the height of the tree does not increase. Therefore, the resulting tree $T'$ is a valid \emph{reduced} DAG elimination tree with maximum depth $d$.
\end{proof}

\begin{figure}[h]
    \centering
\captionsetup{justification=centering}
    \begin{subfigure}[b]{0.35\textwidth}
        \centering
        \includegraphics[width=\textwidth]{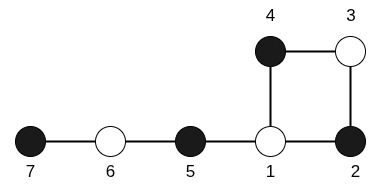}
        \caption{$H_1$}
       \label{fig:h1}
    \end{subfigure}
    \hfill
    \begin{subfigure}[b]{0.45\textwidth}
        \centering
        \includegraphics[width=\textwidth]{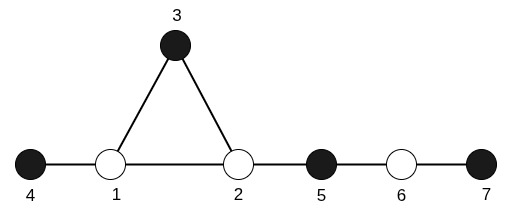}
        \caption{$H_2$}
        \label{fig:h2}
    \end{subfigure}
    
    \caption{$H_1$ and $H_2$ are obstruction for DAG treedepth 2. }
    \label{fig:dag-subfigures}
\end{figure}

Next, we derive the induced minor obstructions for DAG treedepth 0, 1, and 2.

\begin{theorem}\label{thm:dtdobs}
The induced minor obstructions for connected DAGs with DAG treedepth are as follows:
\begin{enumerate}
    \item For DAG treedepth 0: the obstruction is $K_1$.
    \item For DAG treedepth 1: the obstruction is $P_3$.
    \item For DAG treedepth 2: the obstructions are the graphs $C_6, P_7, H_1, H_2$.
\end{enumerate}
\end{theorem}
\begin{proof}
    We can easily verify the DAG treedepth of the listed graphs. We now show that these graphs form the complete set of induced minor obstructions for each treedepth level.

\textbf{Treedepth 0:}
 Any DAG with DAG treedepth one must have at least one source. A single-vertex DAG is trivially an example with one source, and its DAG treedepth is one. In fact, having exactly one source is also a sufficient condition for the DAG treedepth to be one. For example, all DAGs that are cliques (i.e., complete DAGs with a single source) have DAG treedepth one.

\textbf{Treedepth 1:}
 A DAG with treedepth two must have at least two source vertices. In such a case, there must exist at least one non-source vertex that is reachable from both sources. This implies that the path on three vertices, $P_3$, where a single non-source is reachable from two distinct sources, forms a minimal induced minor obstruction for DAG treedepth one.

\textbf{Treedepth 2:}
A DAG with treedepth three must have at least three source vertices. Also, there is no source $s$ such that for each pair of sources $s_i$ and $s_j$, $R(s_i)\cap R(s_j)\subseteq R(s)$. Otherwise, we make $s$ the root and all other sources as children and get a DAG elimination tree of depth two.

Let $T$ be its reduced DAG elimination tree, defined over the set of sources $S$, and let $r$ be the root of $T$. For every source vertex $s \in S$ that is a leaf of $T$ at depth three, let $s_i$ be its parent in $T$. Then there must exist at least one non-source vertex $v \in V \setminus S$ such that: $v \in (R(s) \cap R(s_i)) \setminus R(r)$.

Moreover, for any source vertex $s_1 \in S$ that is a child or grandchild of $r$, there exist a non-source vertex $v$ that is reachable from both $s_1$ and $r$, i.e.,$\exists v \in V \setminus S \ \text{such that} \ v \in R(r) \cap R(s_1)$, otherwise, the graph would not be connected. Otherwise, we can make $r$ the root of a child. There are two cases: 1) the root $r$ has only one child $s$ in the elimination tree $T$, 2) the root $r$ has more than one child.

\begin{itemize}
    \item \textbf{Case 1:} The root $r$ has only one child $s$ in the elimination tree $T$.

    If $s$ has only one child $s_1$ in $T$, then we know that there exists $u\in R(s_1)\cap R(s)$ and $u\notin R(r)$. Also, there exists $v\in R(r)\cap R(s_1)$ and $v\notin R(s)$ otherwise we can root at $s$. Also, there exist $w\in R(r)\cap R(s)$ and $w\notin R(s_1)$. Otherwise, we can root at $s_1$. Thus, the structure induces a $C_6$; otherwise, the DAG treedepth would be at most two. 
    
    If $s$ has more than one child. Let $\{s_1,s_2,\ldots s_l\}$ be the children of $s$ in $T$. So, $R(s_i)\cap R(s_j)\subseteq R(s)$. Given $T$ is reduced elimination tree, there exists $u\in R(s)\cap R(s_i)$ and $u\notin R(r)$. Also, there exists $v\in R(r)\cap R(s_i)$ and $v\notin R(s)$, otherwise, we can root at $s$ and get a DAG elimination tree of depth two. Now, there exist $w$ such that either 
    \begin{itemize}
        \item $w\in R(s)\cap R(s_j)$ and $w\notin R(s_i)$ or
        \item $w\in R(s)\cap R(r)$ and $w\notin R(s_i)$
    \end{itemize}
    Otherwise, we can root at $s_i$ and get an elimination tree of depth two.
    Considering the case $w\in R(s)\cap R(s_j)$ and $w\notin R(s_i)$, we can check that it forms a structure over seven vertices, such that $P_7$ is a subgraph. So, the induced minor would be either $P_7$ or one of its supergraphs, such that the $dtd$ of the supergraph of $P_7$ is greater than two. We can check that those supergraphs are either $H_1$ or $H_2$. So, the induced minor obstructions in this case are $ P_7$, $ H_1$, and $H_2$.

    Now, consider the case $w\in R(s)\cap R(r)$ and $w\notin R(s_i)$, we can check that it forms an induced $C_6$.

    \item \textbf{Case 2:} The root $r$ has more than one child.

     In this case, we first handle the case when one child has a child in $T$. So, there exists $u\in R(s)\cap R(s_1)$ and $u\notin R(r)$. Also, there must exist some non-source in the intersection set of pairwise source vertices other than $s$ that is not reachable by $s$. So, there exist $v\in R(s_1)\cap R(r)$ and $v\notin R(s)$. Now, there exists a non-source in pairwise source reachability intersection; otherwise, we can root at $s_1$ and get a DAG elimination tree of depth two. So, there exist $w$ such that $w\in R(r)\cap R(s')$ and $w\notin R(s_1)$, where $s'$ is another child of $r$ in $T$; or $w\in R(r)\cap R(s)$ and $w\notin R(s_1)$. One can identify $P_7$ or its supergraphs of $dtd$ three as induced minors. 

     Consider the case when more than one child of $r$ has children in $T$. Let $s$ and $s'$ be child of $r$ and $\{s_1,s_2,\ldots s_{l}\}$ are children of $s$ and $\{s'_1,s'_2,\ldots s'_{l'}\}$ are children of $s'$. Since $T$ is reduced DAG elimination tree, there exist $u\in R(s)\cap R(s_i)$ and $u\notin R(r)$. Also, there exist $v \in R(s')\cap R(s_{i'}$ and $v\notin R(r)$. It is possible $u=v$ or $v\in R(s)$. Let $u\neq v$, then we know that there exist $w\in R(s_i)\cap R(r)$ and $w\notin R(s)$. Again, we can check that $P_7$ is a subgraph. Now, consider the case $u=v$, then there exist $w'\in R(s')\cap R(r)$ and $w'\notin R(s)$. Again, we get $P_7$. 
    
    Thus, the induced minor obstructions for treedepth two are $C_6$, $P_7$, and the supergraphs of $P_7$ with $ dtd=3$. We can check that $H_1$ and $H_2$ are such graphs. Hence, the induced minor obstructions of tree depth two are $ C_6$, $ P_7$, $ H_1$, and $H_2$.
\end{itemize}
\end{proof}

In the following two lemmas, we now give relations between the set of minor obstructions of treedepth and the set of induced minor obstructions of DAG treedepth. Let $H_{\text{sub}}$ be the graph obtained by subdividing each edge of $H$ exactly once. Then we have the following two lemmas.
\begin{lemma}\label{lem:indmin1}
    Let $H$ be a minor obstruction for treedepth $k-1$. Then $dtd(H_{\text{sub}}) \geq k$.
\end{lemma}
\begin{proof}
    Let $T$ be an elimination tree of $H$ with depth $k$. Suppose, for contradiction, that there exists a DAG elimination tree $T'$ of $H_{\text{sub}}$ with depth at most $k-1$.

    Note that in $H_{\text{sub}}$, the original vertices of $H$ are treated as source vertices, while each subdivision vertex is a non-source. If $T'$ was valid, we could use the same structure as an elimination tree for $H$ with depth at most $k-1$, contradicting the fact that $H$ is a minor obstruction for treedepth $k$. More concretely, since $H$ contains an edge $\{u, v\}$, there must exist two vertices $u$ and $v$ not on the same root-to-leaf path in $T'$. In $H_{\text{sub}}$, this edge is subdivided by a non-source vertex $w$ such that $w \in R(u) \cap R(v)$ and $w \notin R(s)$ for any other source $s \ne u, v$. This violates the reachability intersection condition of DAG elimination trees, so $T'$ is not valid. Hence, no DAG elimination tree of $H_{\text{sub}}$ can have depth less than $k$, therefore $dtd(H_{\text{sub}}) = k$.
\end{proof}

\begin{lemma}\label{lem:indmin2}
    Let $H$ be an induced minor such that $dtd(H) = k$. Then, there exists a sequence of edge contractions between source and non-source vertices such that the resulting graph $H'$ is a minor of treedepth $k$.
\end{lemma}
\begin{proof}
     We observe that after performing edge contractions, the remaining vertices correspond exactly to the set of source vertices in $H$. Let $T$ be a DAG elimination tree of $H$ with depth $k$. For each non-source vertex $u$ that is reachable from multiple source vertices, we contract $u$ with the source vertex that is an ancestor of all other sources reaching $u$ in $T$. We claim that the resulting graph $H'$, obtained after these edge contractions, has treedepth exactly $k$. 
    
    Suppose, for contradiction, that $H'$ admits an elimination tree $T'$ of depth at most $k-1$. Now consider using the same tree structure $T'$ as a DAG elimination tree for $H$, where the contracted vertices are replaced back with their original non-source vertices. Consider any two sibling source vertices $s_i$ and $s_j$ in $T'$. For each non-source vertex $u \in R(s_i) \cap R(s_j)$, by construction, $u$ was contracted into some source $s$ that is an ancestor of both $s_i$ and $s_j$ in $T$. Hence, $u \in R(s)$, contradicting the assumption that $dtd(H) = k$. Therefore, such a $T'$ of depth $k-1$ cannot exist, and we conclude that the treedepth of $H'$ is exactly $k$.
\end{proof}

Using \Cref{lem:indmin1}, we observe that the subdivision of any minor obstruction for treedepth \(k\) is an induced-minor obstruction for DAG treedepth \(k\). Conversely, by \Cref{lem:indmin2}, from any induced-minor obstruction for DAG treedepth \(k\), one can obtain a minor obstruction for treedepth \(k\) via a sequence of edge contractions.

Let \(S\) denote the set of minor obstructions for treedepth \(k\). For each \(H \in S\), we construct a corresponding pattern \(H_{\mathrm{sub}}\). By construction, each \(H_{\mathrm{sub}}\) is an induced-minor obstruction for DAG treedepth \(k\). Moreover, we must also consider all supergraphs of \(H_{\mathrm{sub}}\). Let \(H^{*}\) be a supergraph of \(H_{\mathrm{sub}}\) such that \(\mathrm{dtd}(H^{*}) \ge k\). Then \(H^{*}\) is also an induced-minor obstruction for DAG treedepth \(k\).

Therefore, the set of induced-minor obstructions for DAG treedepth \(k\) can be derived from the set of minor obstructions for treedepth \(k\). Since the number of minor obstructions for treedepth \(k\) is $\frac{(2^{2^{k-1}-k})(1 + 2^{2^{k-1}-k})}{2}$ \cite{giannopoulou2009obstructions}, it implies that the number of induced-minor obstructions for DAG treedepth \(k\) is at least $\frac{(2^{2^{k-1}-k})(1 + 2^{2^{k-1}-k})}{2}$. 

Giannopoulou and Thilikos \cite{giannopoulou2009obstructions} have listed the obstruction set of treedepth up $3$. The number of graphs in the obstruction set of treedepth $3$ is twelve. Consequently, to obtain the induced minor obstructions, one must consider the subdivisions of these graphs together with all of their supergraphs. Therefore, we gave induced minor obstruction till DAG treedepth at most $2$.

\subsection{Conditional Hardness and Small Pattern Counting}\label{sec:dtdlower}

In this section, we present conditional lower bounds for counting patterns in cubic time using constant space. We then show that all patterns on at most nine vertices can be counted in cubic time and constant space. Finally, we exhibit a specific pattern on ten vertices and prove a conditional lower bound, showing that it cannot be counted in \(O(n^{4-\varepsilon})\) time and constant space.

We have observed that DAGs with $\mathrm{dtd}=1$ are precisely those having a single source. Moreover, specific orientations, such as the star graph with all leaves as sources, may contain multiple sources but still admit a linear-time, constant-space algorithm for counting homomorphisms. More generally, if a DAG admits an orientation where either the number of sources or non-sources is one, then homomorphism counting remains tractable in linear time and constant space. The set of patterns with $dtd=1$ is cliques and star graphs. Leveraging the framework of \cite{curticapean2017homomorphisms}, we extend our result for subgraph counting. One can observe that the spasm of cliques and star graphs is a clique or a star graph and therefore has $dtd$ one. Hence, one can count the number of star graphs and cliques appearing as subgraphs in linear time and constant space. However, for the path $P_4$, an alternating edge orientation yields two sources and two non-sources, suggesting a potential boundary of tractability. This leads us to the following conjecture:

\begin{restatable}{conjecture}{conj}
Let $d$ be a constant. For any constant $\varepsilon > 0$, there is no $O(n^{2-\varepsilon})$ time constant-space algorithm for counting $\mathrm{hom}(P_4, H)$, where the input graph $H$ is a $d$-degenerate graph.
\end{restatable}

Assuming the above conjecture, we now give a characterization of patterns that appear as subgraphs and can be counted in linear time and constant space.

\begin{theorem}\label{thm:lowerdtd1}
Patterns that are countable as subgraphs in linear time and constant space are precisely the graphs with $\operatorname{dtd}=1$ and the star graphs.
\end{theorem}
\begin{proof}
    A DAG with $\operatorname{dtd}=1$ has exactly one source vertex. 
Such graphs are either cliques or stars. 
Moreover, every graph in the spasm of a clique or a star also has $\operatorname{dtd}=1$. 
Hence, for each vertex in the spasm of a $\operatorname{dtd}=1$ graph, the number of homomorphisms can be counted in linear time and constant space. 
Consequently, the number of subgraphs of a DAG with $\operatorname{dtd}=1$ can also be computed in linear time and constant space.

Furthermore, there exist orientations of star graphs in which the number of sources is greater than one. However, in such cases, the number of non-sources is exactly one. 
Therefore, counting star subgraphs only requires computing the degree of each vertex in the host graph, which can be done in linear time and constant space.
\end{proof}

Next, to characterize patterns that can be counted in cubic time and constant space, we assume the following conjecture for a general graph.

\begin{conjecture}\label{conj:triangle}
For any constant $\varepsilon > 0$, there is no $O(n^{3-\varepsilon})$-time and constant-space algorithm for counting $C_3$ (triangles) in an arbitrary input graph $G$.
\end{conjecture}

It is shown in \cite{bera2020linear} that the problem of counting $C_6$ even in degeneracy two is equivalent to counting triangles in a general graph. Thus, assuming  \Cref{conj:triangle}, there is no $O(n^{3-\varepsilon})$-time and constant-space algorithm for counting $C_6$ in a graph $G$ with degeneracy equal to two. Further, consider the obstruction set of graphs with $\operatorname{dtd}=2$. It can be verified that the spasm of every graph appearing in this obstruction set contains a triangle $C_3$ when reduced in a manner similar to~\cite{bera2021near}. 
Hence, we have the following :

\begin{theorem}
There is no $O(n^{3-\varepsilon})$-time and constant-space algorithm for counting graphs with $\operatorname{dtd}=3$ in a graph $G$ with degeneracy equal to two.
\end{theorem}

Next, we prove that every pattern graph on at most nine vertices has DAG treedepth at most three. Consequently, by applying the framework of~\cite{curticapean2017homomorphisms}, it follows that homomorphisms, subgraphs, and induced subgraphs of all such patterns can be counted in \(O(n^3)\) time using constant space.

\begin{theorem}\label{thm:9vertices}
Let $d$ be a constant. For all patterns with at most nine vertices, we can count the number of occurrences as induced subgraphs in $O(n^3)$ time and constant space for an $n$-vertex $d$-degenerate graph given as input.
\end{theorem}
\begin{proof}
    Assume that the pattern is connected.  
Consider a DAG with at most three sources. 
It is easy to observe that the \textup{dtd} of such a DAG is bounded by three.

Now consider the case when the number of non-source vertices is at most three. 
If the number of non-sources is at most two, then we can construct an elimination tree by taking one source vertex as the root and another source vertex reaching the other non-source vertex, while all remaining source vertices can be attached as leaves. 
If the number of non-sources is three, since the graph is connected, there must exist a source vertex that reaches at least two non-source vertices. 
We then select one more source vertex that reaches the remaining non-source vertex. 
Thus, we obtain an elimination tree of height at most three.

Next, consider the case where the DAG has four source vertices. 
The $\operatorname{dtd}$ of this DAG equals four if and only if there exist unique reachable non-source vertices for every pair of source vertices. 
Hence, such a DAG must contain at least six non-source vertices.

Now consider the case when the number of source vertices is five. 
Then the number of non-source vertices is four. 
Let $s$ be a source vertex with the maximum reachability set $R(s)$, so $|R(s)| \ge 2$. 
If there exists another source vertex $s'$ such that $|R(s)\cup R(s')| = 4$, then we again obtain an elimination tree of height three. 
Otherwise, for every remaining source vertex $s'$, we have $|R(s)\setminus R(s')|\le 1$. 
In this case, we make $s$ the root of the elimination tree and attach one source vertex as a child corresponding to each remaining non-source vertex. 
Thus, we again obtain an elimination tree of height at most three. 

Therefore, every DAG with at most nine vertices has $\operatorname{dtd}\le 3$. 
Since the spasm and all supergraphs of any graph with at most nine vertices also have at most nine vertices, 
we can count the induced subgraphs of all patterns with at most nine vertices in $O(n^3)$ time and constant space.
\end{proof}

Using the above theorem and the fact that $dtd(P_{10})=3$, we can conclude that \begin{corollary}
    We can count subgraph $P_{10}$ in time $O(n^3)$ using constant space.
\end{corollary}

Using \Cref{thm:9vertices}, we get an $O(n^3)$ time constant space algorithm up to nine vertices. We now present the hardness results for ten vertices. Consider a DAG on ten vertices, which is one subdivision of $k_4$. Using the reduction shown in the proof of \Cref{con:paul}, we get that it reduces to counting $k_4$ in a general graph. 
Now, assuming the conjecture that no $O(n^{4-\varepsilon})$-time and constant-space algorithm exists, for counting $k_4$ in general graphs. We get that there is no $O(n^{4-\varepsilon})$-time and constant-space algorithm for counting $k_4$ with one subdivision in bounded degenerate graphs.

It is worth noting that efficient algorithms for counting cliques exist using combinatorial techniques or fast matrix multiplication. 
However, both these approaches require non-constant space. 
To achieve constant-space algorithms, one must rely on divide-and-conquer strategies based on tree depth or matched tree depth. 
As shown earlier, $\operatorname{dtd}$ is bounded above by tree-depth. 
For sparse patterns, matched tree-depth can lead to faster algorithms. 
For example, all patterns with at most eleven edges can be counted in $O(n^3)$ time and constant space. 
Consequently, as a corollary, the paths $P_{12}$ and cycles $C_{11}$ can also be counted as subgraphs in $O(n^3)$ time and constant space. 
However, for dense patterns and induced subgraph counting, DAG treedepth provides a more advantageous framework compared to previously known constant-space algorithms.
\section{Constant Space computation of $\mathrm{hom}(H,G), \mathrm{ sub}(H,G$ and $\mathrm{ind}(H,G)$}\label{sec:dtdconstant} In this section we show that we can compute $\mathrm{ind}(H,G)$ in time $O(n^{\lfloor\frac{k}{4}\rfloor+2})$ and $\mathrm{hom}(H,G)$, $\mathrm{sub}(H,G) $ in time $\min(O(n^{\lfloor\frac{k}{4}\rfloor+2}),O(n^{\lfloor\frac{l}{5}\rfloor+3}))$ using $O(1)$-space, where $|V_H|=k$, $|E_H|=l$ and, $n$ is the number of vertices in $G$. To this end, we prove the following theorem.

\begin{theorem}\label{thm:dtdk/4}
 Let $H$ be a pattern DAG with $k$ vertices. Then, $dtd(H) \leq \lfloor\frac{k}{4}\rfloor + 2$.   
\end{theorem}
\begin{proof}
We make a DAG elimination tree of depth at most $\lfloor\frac{k}{4}\rfloor + 2$. To this end, we add source vertices that maximize the number of discovered vertices.
   We incrementally construct a root-to-leaf path $T_p$ of a DAG elimination tree. Let $D$ be the set of discovered (processed) vertices, initialized as empty. At each step, select a source vertex $s$ such that $|R(s) \setminus R(D)|$ is maximized. We add $s$ as a child in $T_p$.

\begin{itemize}
    \item \textbf{Case 1:} $|R(s_i) \setminus R(D)| \geq 3$
    
Adding such a source $s_i$ increases the depth by one but makes at least four vertices discovered (one source and three non-sources). Since all subsequent sources are added below $s_i$, the reachability condition is preserved.

Now, consider the residual graph $G'$ after deleting all sources added so far and their reachable vertices. If $|R(s_i)| = 0$, then add $s_i$ as a child to the last vertex added in $D$. This increases depth by at most one. If $|R(s_i)| = 1$, say $R(s_i) = \{u\}$ then  if a source $s$ already in $D$ reaches $u$, attach $s_i$ under $s$. Otherwise, attach $s_i$ under the last added vertex in $D$, and attach any other sources that also reach only $u$ as its children.

In all these cases, the depth increases by at most two. We now reduce the graph to a residual subgraph where each source reaches exactly two non-sources, and each non-source is reached by at least two sources.

\item \textbf{Case 2:} $|R(s_i) \setminus R(D)| = 2$ and $|P(u) \setminus D| \geq 2$

If there exists a source $s$ such that $R(s) \cap R(s_i) = \emptyset$ for all $s_i$, we can add $s$ as a child to the last source in $D$, and attach all sources whose reachable sets are contained in $R(s)$ under $s$. This increases depth by at most 2.

Now suppose that every source $s_i$ has some overlap $R(s_i) \cap R(s_j) \ne \emptyset$ with another source $s_j$. Then, adding $s_i$ removes at least four vertices from the graph: the source $s_i$ and two non-sources shared with other sources. If any leftover source $s_j$ now reaches only one undiscovered vertex, we can handle it as in Case 1.

\end{itemize}

In every case, we discover at least 4 vertices when we increase the DAG treedepth by one. Also, for corner cases when $|R(s)|\leq 1$, we can handle them by increasing the treedepth by two. Hence, the depth of the elimination tree of the DAG is bounded by $\lfloor\frac{k}{4}\rfloor + 2$. Therefore, $dtd(H) \leq \lfloor\frac{k}{4}\rfloor + 2$.
 
\end{proof}

Using \Cref{thm:dtdalgo} and \Cref{thm:dtdk/4}, we get the following result.
\begin{theorem}\label{thm:dtdvertices}
  Consider any $k$-node pattern graph $H = (V_H, E_H)$. Then one can compute $\mathrm{hom}(H,G),\mathrm{sub}(H,G)$ and $\mathrm{ind}(H,G)$ in $O(f(k,d) \cdot n^{\lfloor\frac{k}{4}\rfloor+2})$ time using $O(1)$-space.
\end{theorem}
The above theorem, keeping the time complexity the same, improves the space complexity of \cite{bressan2021faster} from polynomial to a constant for counting all patterns of size $k$.

\begin{theorem}\label{thm:dtdedgebound}
 Let $H$ be a DAG with $l$ edges. Then, $dtd(H) \leq \lfloor\frac{l}{5}\rfloor + 3.$

\end{theorem}

\begin{proof}
   We follow a similar strategy to that in  \Cref{thm:dtdk/4}. The key idea is to show that for each unit increase in depth, we can safely eliminate at least 5 edges from the graph. The additive constant accounts for small residual cases.

As discussed before, after handling sources with low reachability (e.g., reaching at most one non-source), we reduce the graph to a simpler form where:
\begin{itemize}
    \item Every non-source vertex is reachable by at least two sources.
    \item For any two sources $s_i$ and $s_j$, $R(s_i) \nsubseteq R(s_j)$.
    \item For any two non-sources $u$ and $v$, $P(u) \nsubseteq P(v)$.
\end{itemize}

When we add a source vertex $s$ to the elimination path $T_p$, we eliminate all edges adjacent to its reachable non-sources $u \in R(s)$. Thus, the number of eliminated edges is at least $\sum_{u \in R(s)} \deg(u)$.

\begin{itemize}
    \item \textbf{Case 1:} $|R(s_i) \setminus R(D)| \geq 3$
    
Suppose $s_i$ reaches distinct non-sources $u_1, u_2, u_3$, and for each pair $u_j \neq u_k$, the sets $P(u_j) \setminus P(u_k)$ are non-empty as $P(u_j)\nsubseteq P(u_k)$. This implies that each non-source $u_i$ contributes at least 2 edges, where one edge is between the source vertex $s_j$ and $u_i$ and one edge between $u_i$ and some source $s$. Note that such a source exists because $P(u_i)\geq 2$. Also, if $|P(u_i)|=2$, then $P(u_i)\cap P(u_j)=\{s_i\}$ otherwise, $P(u_i)\subseteq P(u_j)$. Therefore, we eliminate at least 6 edges by adding $s_i$, increasing the depth by 1.

\item \textbf{Case 2:} $|R(s_i) \setminus R(D)| = 2$

\begin{itemize}
    \item \textbf{Case A:} There exists a non-source $u \in R(s_i)$ with $|P(u)| \geq 3$
    
  Let $s_i$ reach $u_1$ and $u_2$, and assume $|P(u_1) \setminus P(u_2)| \geq 1$ otherwise $P(u_1)\subseteq P(u_2)$. Let $|P(u_2)|\geq 3$ and $|P(u_1)|\geq 2$. So, $|P(u_1)\cup P(u_2)|$ is at least $4$. Now, if $|P(u_1)\cup P(u_2)|$ is at least five, then we can say that adding only one vertex $s_i$, we eliminate at least five edges. Now, consider the case $|P(u_1)\cup P(u_2)|=4$. It is only possible when $|P(u_1)|=2$. So, consider source $s\in P(u_1)$ other than $s_i$. We can check that $|R(s)\setminus R(D)|= 1$. We can eliminate source $s$ after adding $s_i$ to the path and increasing the depth by one. $s$ can be handled as described above using the additive factor two when the source reaches at most one undiscovered non-source. Thus, in this step as well, we increased the DAG treedepth by one and eliminated five edges.

  \item \textbf{Case B:} Every non-source $u$ has $|P(u)| = 2$
  
  In this scenario, the residual graph consists only of cycles or disjoint cycles. For a cycle with $l'$ edges, its DAG treedepth is at most $\log(\lceil l'/2 \rceil)$. Since $\frac{l'}{5} + 1 \geq \log(\lceil l'/2 \rceil)$ for all $l'$, the depth bound still holds.
\end{itemize}
\end{itemize}

Thus, in all cases, we can safely remove at least 5 edges per unit increase in depth (with an additive constant), which gives us that $\operatorname{dtd}(H) \leq \lfloor\frac{l}{5}\rfloor + 3$.
 
\end{proof}

Again, using \Cref{thm:dtdalgo} and \Cref{thm:dtdedgebound}, we have the following result.
\begin{theorem}\label{thm:dtdedges}
  Consider any $l$-edge pattern graph $H = (V_H, E_H)$. Then one can compute $Hom(H,G)$ and $Sub(H,G)$ in $O(f(k,d) \cdot n^{\lfloor\frac{l}{5}\rfloor+3})$ time using $O(1)$-space.
\end{theorem}

\begin{remark}
  From \Cref{thm:dtdvertices} and \Cref{thm:dtdedges}, we can compute $\mathrm{hom}(H,G)$ and $\mathrm{sub}(H,G)$ in time $\min(O(f(k,d) \cdot n^{\frac{k}{4}+2}),O(f(k,d) \cdot n^{\frac{l}{5}+3}))$ using $O(1)$-space.  
\end{remark}

In \Cref{thm:dtdedgebound}, we observed that in Case 1, we can achieve a bound of $\frac{\ell}{6} + 2$ on the DAG treedepth. For the remaining cases, we can construct an undirected graph $G$ as described in ~\Cref{lem:boundedreachablity}, where the vertex set consists of non-source vertices from the residual DAG $G'$. Importantly, the number of edges in $G$ is equal to the number of source vertices in $G'$, and the number of edges in $G'$ is exactly twice the number of its source vertices.

Now, suppose there exists a general upper bound on treedepth for undirected graphs in terms of the number of edges, i.e., $td(G) \leq \frac{m}{\alpha}$. Then, applying this to the graph $G$, we obtain: $dtd(G') \leq \frac{m}{2\alpha}.$ This  leads to the following conditional corollary:

\begin{corollary}
If for any undirected graph $G$ with $m$ edges, $td(G) \leq \frac{m}{3},$ then for any pattern DAG $H$ with $\ell$ edges, $dtd(H) \leq \frac{\ell}{6} + 2.$
 \end{corollary}

Next, we deal with DAG treewidth. Since treewidth is a well-studied graph parameter and efficient approximation algorithms exist for computing it \cite{bonnet2025treewidth}, it is reasonable to bound other parameters in terms of treewidth. Therefore, in the next section, we first define DAG treewidth \cite{bressan2021faster} and obtain many relationships between the DAG treewidth and the treewidth of graphs in $\mathcal{G}$.

\section{DAG Treewidth}\label{sec:treewidth}

In this section, we first establish several fundamental properties of DAG treewidth. Building upon these we prove that for any DAG $H$ with $k$ vertices, $\mathrm{dtw}(H) \leq \frac{k}{5} + 3$, and leverage this bound to derive improved runtime for computing $\mathrm{hom}(H, G)$, $\mathrm{sub}(H, G)$, and $\mathrm{ind}(H, G)$. We further prove a conjecture recently posed by \cite{paul2024subgraph} on the 10-vertex pattern graph $H$. Along the way, we obtain several additional structural bounds which may be of independent interest. To this end, we first define the DAG treewidth.

\begin{definition}(DAG Tree Decomposition)\label{def:dagtreedecomposition}
    Let  $\vec{H}$ be a DAG with a set of sources $S$. A \emph{DAG tree decomposition} of $\vec{H}$ is a rooted tree $T = (\mathcal{B}, E)$, where each node $B \in \mathcal{B}$ is a \emph{bag} consisting of a subset of sources from $S$, i.e., $B \subseteq S$, and the following properties are satisfied:

    \begin{itemize}
    \item \textbf{Coverage:} Every source appears in at least one bag, i.e.,  $\bigcup_{B \in \mathcal{B}} B = S $.
    
    \item \textbf{Reachability Intersection Property:} For any three bags $B, B_1, B_2 \in \mathcal{B} $, if $B$ lies on the unique path between $B_1$ and $B_2$ in the tree $T$, then
    $$\mathrm{R}_H(B_1) \cap \mathrm{R}_H(B_2) \subseteq \mathrm{R}_H(B).$$
\end{itemize}
\end{definition}

\begin{definition}(DAG Treewidth)
Let the \emph{DAG width} of a DAG tree decomposition $T = (\mathcal{B}, E)$ be defined as the size of the largest bag in $T$, i.e., the maximum number of sources in any bag. Then, the \emph{DAG treewidth} of a directed acyclic graph $\vec{H}$ is the minimum DAG width over all possible DAG tree decompositions of $\vec{H}$. For an undirected graph $H$, the DAG treewidth of $H$ is the maximum of the DAG treewidth of all acyclic orientations of $\vec{H}$. We denote the DAG treewidth of a graph $\vec{H}$ by $dtw(\vec{H})$.

\end{definition}

\begin{observation}\label{obs:uniqueintersection}
    Let $\vec{H}$ be a directed acyclic graph (DAG) with a fixed orientation. Suppose there are two sources $u$ and $v$ such that there exists a vertex $w \in \mathrm{R}_H(u) \cap \mathrm{R}_H(v)$, and $w \notin \mathrm{R}_H(s_i)$ for any other source $s_i \neq u, v$. Then it is easy to see that, in any DAG tree decomposition of $\vec{H}$, the sources $u$ and $v$ must appear in the same bag or in adjacent bags.
\end{observation}

\begin{observation}\label{obs:equaltreewidth}
    As edge contractions do not increase treewidth, we have $\mathrm{tw}(G_S) \leq \mathrm{tw}(\textsc{Bip}(\vec{H}))$.
\end{observation}

Let $T = (\mathcal{B}, E)$ be a DAG tree decomposition. We denote by $B_{leaf}$ a leaf node of the tree $T$. For any node $B \in \mathcal{B}$, we denote its parent by \textsc{Parent}$(B)$. We denote the root of the tree $T$ by $\textsc{Root}(T)$. For any two nodes $B_1, B_2 \in \mathcal{B}$, the unique path between $B_1$ and $B_2$ in $T$ is defined as $T(B_1,B_2)$.

\subsection{Properties of DAG Treewidth}\label{sec:dtwproperty}

In this section, we establish several important properties of treewidth and DAG treewidth ($\mathrm{dtw}(\vec{H})$), along with important relationships between them. By definition, the DAG treewidth of an undirected graph is the maximum DAG treewidth over all its acyclic orientations. Consequently, results derived for DAGs can be naturally extended to undirected graphs by considering all possible acyclic orientations. Therefore, any property proven for a fixed acyclic orientation of a DAG applies to the corresponding undirected graph as well. However, in ~\Cref{dtwinducedminor} (which contracts a single edge, reducing the number of vertices by one), \Cref{lem:twbound} (which compares the treewidth of undirected graphs), and \Cref{lem:onesubdivision} (which introduces a new vertex on every edge (known as a subdivision vertex)), it is necessary to account for all acyclic orientations. For this reason, we present the proofs of these lemmas in the undirected setting.

Treewidth is a well-studied structural parameter that plays a central role in the design of dynamic programming algorithms for many NP-complete problems on graphs of bounded treewidth. In particular, for a pattern graph $H$ with treewidth $t$, the number of homomorphisms from $H$ to a host graph $G$ with $n = |V(G)|$ can be computed in time $n^{O(t)}$.

Bressan~\cite{bressan2021faster} introduced DAG treewidth to enable near-linear time algorithms for counting homomorphisms from certain graphs, such as cliques and complete bipartite graphs, whose treewidth is unbounded, but whose DAG treewidth is just 1. This highlights DAG treewidth as a powerful and relevant parameter for homomorphism counting. However, here we derive a hardness result for DAG treewidth. To this end, we show the following lemma.

\begin{lemma}\label{lem:dtw=ghw}
Let $H$ be a DAG with the set of sources $S$ and non-sources $T$. 
For each source $s \in S$, let $R(s) \subseteq T$ denote the set of vertices reachable from $s$. 
Construct a hypergraph $\mathcal{H}$ whose vertex set is $T$ and whose hyperedge set is 
$\mathcal{E} = \{ R(s) \mid s \in S \}$. 
Then the DAG treewidth of $H$ equals the generalized hypertree width of $\mathcal{H}$, i.e.,
\[
\operatorname{dtw}(H) = \operatorname{ghw}(\mathcal{H}).
\]
\end{lemma}

\begin{proof}
(\emph{dtd} $\Rightarrow$ \emph{ghd}) 
Let $(T_\tau, \{\beta(t)\}_{t \in V(T_\tau)})$ be a DAG tree decomposition of $H$. 
For each node $t$, construct a GHD node with 
$\chi(t) = \beta(t) \cap T$ and 
$\lambda(t) = \{ R(s) : s \in \beta(t) \cap S \}$. 
The coverage and connectedness conditions of a GHD follow directly from those of the DTD. 
Since each source $s$ corresponds to a single hyperedge $R(s)$, the width (measured by the number of sources/hyperedges) is preserved. 
Hence, $\operatorname{ghw}(\mathcal{H}) \le \operatorname{dtw}(H)$.

(\emph{ghd} $\Rightarrow$ \emph{dtd}) 
Conversely, let $(T_\tau, \{(\chi(t), \lambda(t))\}_{t \in V(T_\tau)})$ be a generalized hypertree decomposition of $\mathcal{H}$. 
We obtain a directed tree decomposition of $G$ by replacing each hyperedge $R(s) \in \lambda(t)$ with its corresponding source $s$. 
Formally, define $\beta(t) = \chi(t) \cup \{ s \in S : R(s) \in \lambda(t) \}$. 
The coverage and connectedness properties of a DTD follow from those of the GHD, and the width is again preserved. 
Hence, $\operatorname{dtw}(H) \le \operatorname{ghw}(\mathcal{H})$. Combining both directions yields $\operatorname{dtw}(H) = \operatorname{ghw}(\mathcal{H})$.
\end{proof}

To obtain a hardness result for DAG treewidth, we use the result of $\cite{fischl2018general}$, which shows that $ghw(\mathcal{H}\le 2)$ is NP-complete. To this end, we have the following result.

\begin{theorem}\label{thm:2np}
    Let $\vec{H}$ be a DAG, then deciding whether $dtw(H)\le 2$ is NP-complete.
\end{theorem}
\begin{proof}
    The proof follows using \Cref{lem:dtw=ghw} and the fact that deciding whether $ghw(\mathcal{H}\le 2)$ is NP-complete.
\end{proof}

The following lemma shows that $\mathrm{dtw}(\vec{H}) \leq \mathrm{tw}(H) + 1$. This bound suggests that the DAG treewidth of a graph is never significantly worse than its treewidth.

\begin{lemma}\label{lem:twbound}
    For any undirected graph $H$, $\vec{H}$ is any acyclic orientation of $H$. Then, $\mathrm{dtw}(\vec{H}) \leq \mathrm{tw}(H) + 1.
$
\end{lemma}

\begin{proof}
 Let $T$ be a tree decomposition of $H$ where $\mathrm{tw}(H) = k$. Fix any arbitrary acyclic orientation of $\vec{H}$. Observe that every source in $\vec{H}$ must appear in some bag of $T$, since the decomposition covers all vertices.

Now, consider any pair of source vertices $s_i$ and $s_j$ such that their reachability sets intersect, i.e., $R(s_i) \cap R(s_j) \neq \emptyset$. Suppose $s_i \in B_1$ and $s_j \in B_2$, where $B_1$ and $B_2$ are bags in $T$. Let $T(B_1, B_2)$ denote the unique path between $B_1$ and $B_2$ in the tree decomposition. If $R(s_i) \cap R(s_j) \subseteq R(B)$ for all intermediate bags $B \in T(B_1, B_2)$, then the reachability condition is already satisfied. Otherwise, let $u \in R(s_i) \cap R(s_j)$ be a vertex that is not contained in any such intermediate bag.

Since $u \in R(s_i) \cap R(s_j)$, there exists a directed path from $s_i$ and $s_j$ to $u$, consisting only of non-source vertices:

$$
s_i \rightarrow u_1 \rightarrow \cdots \rightarrow u \leftarrow u'_1 \leftarrow \cdots \leftarrow u'_l \leftarrow s_j.
$$

Therefore, there is a corresponding path in the tree decomposition:

$$
B_1 - B'_1 - B'_2 - \cdots - B'_t - B_2
$$
where each bag $ B'_i$ must contain at least one vertex from the above path, by the property of the tree decomposition. In each such bag $ B'_i$, we replace one non-source vertex (e.g., $u_k$ or $ u'_k$) with the source vertex $s_i$. This replacement does not increase the size of the bag. By applying this process iteratively to all these source pairs $s_i, s_j$, we can ensure that the reachability condition for a DAG tree decomposition is satisfied throughout the tree. Thus, $dtw(\vec{H})\leq tw(H)+1$.

\end{proof}

The treewidth ($\mathrm{tw}$) of a graph remains invariant under any number of edge subdivisions. In many graph classes with bounded treewidth, the number of vertices can far exceed the treewidth. It is also known that cliques and quasi-cliques typically exhibit high treewidth but low DAG treewidth. However, even a single edge subdivision in such graphs can lead to a significant increase in the DAG treewidth. For example, $tw(K_n)=n-1$ and $dtw(K_n)=1$ but with only only subdivision we get that $tw(K_{n_{sub}})=n-1$ whereas $dtw(K_{n_{sub}})=\lceil \frac{n}{2}\rceil$.

To account for this behavior, ~\Cref{lem:onesubdivision} establishes an upper bound on the DAG treewidth following a single subdivision per edge. This result is especially useful for analyzing the DAG treewidth of cliques and quasi-cliques under limited subdivision. By combining ~\Cref{lem:onesubdivision} with ~\Cref{lem:twbound}, one can obtain a sharper upper bound on DAG treewidth in such settings.

\begin{lemma}\label{lem:onesubdivision}
    Let $\gmain$ be an undirected graph with $n$ vertices. Let $G_{\text{sub}}$ be a graph obtained by subdividing each edge of $\gmain$ at most once. Then, $dtw(G_{\text{sub}})\le \lceil\frac{n}{2}\rceil$.
\end{lemma}
\begin{proof}
Let $O_i$ be arbitrary but fixed acyclic orientation of $G_{\text{sub}}$ and let $S = \{s_1, s_2, \dots, s_k\}$ be the set of source vertices in $G_{\text{sub}}$. These source vertices may be original vertices of $\gmain$ or subdivision vertices introduced in $G_{\text{sub}}$.

If $|S|\le n$, then we can partition the source set $S$ into two bags such that each bag contains at most $\lceil \frac{n}{2}\rceil$ vertices.

Now consider the case when $|S| > n$. This means some of the subdivision vertices in $G_{\text{sub}}$ are also source vertices. We construct a DAG tree decomposition $T$ as follows: Initialize $T$ with two bags, $B_1$ and $B_2$. We also maintain two corresponding lists $\SEEN(B_1)$ and $\SEEN(B_2)$, which keep track of the vertices from $\gmain$ added so far.

\paragraph*{Step 1}
Identify all source vertices that are also original vertices of $\gmain$. Add them alternately to $B_1$ and $B_2$. Each time we add such a vertex to a bag, we also add it to the corresponding $\SEEN$ list. Note that if a vertex $v$ from $\gmain$ is a source in $G_{\text{sub}}$, then none of the subdivision vertices adjacent to $v$ can be a source.

\paragraph*{Step 2}
Handle source vertices that are subdivision vertices. For a subdivision vertex $uv$, which comes from subdividing edge $\{u, v\}$ in $\gmain$:
\begin{itemize}
    \item \textbf{Step A:} If both $u$ and $v$ are not present in either $\SEEN(B_1)$ or $\SEEN(B_2)$, add $uv$ to both $B_1$ and $B_2$. Also, add $u$ and $v$ to both $\SEEN(B_1)$ and $\SEEN(B_2)$.

    \item \textbf{Step B:} Now look for a source vertex $vv_i$ where $v_i$ is not yet in either $\SEEN(B_1)$ or $\SEEN(B_2)$. Add $vv_i$ to the smaller of $B_1$ or $B_2$ and update the corresponding $\SEEN$ list by adding $v_i$. Continue this process: alternate between $B_1$ and $B_2$, adding new source vertices and updating $\SEEN$ lists accordingly.

If you reach a subdivision vertex $v_pv_q$ such that no further extension is possible (i.e., there is no next source $v_qv_r$ with $r$ not in $\SEEN$), then go back to Step A and repeat the process with a fresh $v_sv_t$ where both endpoints are still unseen.

\end{itemize}
Through these procedures, we ensure that (almost) all vertices of $\gmain$ are eventually added to $\SEEN(B_1)$ or $\SEEN(B_2)$, except potentially the vertex $v_q$ where Step B terminates. Moreover, no further source $v_qv_r$ exists with $r$ outside both $\SEEN$ lists.

\paragraph*{Step 3}
Now handle the remaining subdivision source vertices $v_sv_t$ where $v_s$ is already in some $\SEEN$ list, and $v_t$ is not. Add $v_sv_t$ to the smaller of $B_1$ or $B_2$ and update the respective $\SEEN$ list by including $v_t$. Repeat this until all vertices from $\gmain$ are in one of the $\SEEN$ lists.

After Steps 1–3, we claim that both bags have size at most $\lceil\frac{n}{2}\rceil$. Why? Because each time we add a source vertex, we are effectively marking a new vertex from $\gmain$ as $\SEEN$. In Step 2A, while a source vertex is added to both bags, it causes two $\gmain$ vertices to be added to each $\SEEN$ list. Additionally, we always ensure the bag sizes differ by at most 1, so the maximum size of each bag remains within $\lceil\frac{n}{2}\rceil$.

Now, there may still be some remaining source vertices $u_lv_l$ such that both $u_l$ and $v_l$ have already been added to the $\SEEN$ lists. For these vertices, we make the following claim:

\begin{claim}
    Each remaining such source vertex $u_lv_l$ has both endpoints in the same $\SEEN$ list, either $\SEEN(B_1)$ or $\SEEN(B_2)$.
\end{claim}

\begin{proof}
   Suppose $u_l$ is in $\SEEN(B_1)$ and $v_l$ is in $\SEEN(B_2)$. Then $u_l$ must have been added through Step 2 or 3. If it was added in Step 2 (except at the stopping point), then it must have been added to both $\SEEN$ lists. If $u_l$ were the terminal vertex of Step 2B, then we would still be able to process $u_lv_l$ in the next step, contradicting the assumption that it is a leftover. Similarly, if both $u_l$ and $v_l$ were previously unseen, they would have been picked up by Step 2A. So, either both endpoints were added together, or they ended up in the same $\SEEN$ list through the process. 
\end{proof}

\paragraph*{Step 4}
For all such remaining source vertices $u_lv_l$, attach a \emph{leaf bag} to the corresponding bag $B_i$ (either $B_1$ or $B_2$), containing just the vertex $u_lv_l$. We attach it to $B_i$ such that both $u_l$ and $v_l$ are in $\SEEN(B_i)$.

Now, we argue that the DAG tree decomposition $T$ is valid. All sources are included in some bag. For the reachability condition, observe that for a subdivision vertex $uv$, we have $\mathrm{R}_{T}(uv) \subseteq \mathrm{R}_{T}(u) \cap \mathrm{R}_{T}(v)$. (Here, we are slightly abusing notation since $u$ and $v$ may not be sources themselves.)

All subdivision sources added in Steps 1–3 are in the same or adjacent bags, so reachability is trivially preserved. For those added in Step 4, the claim ensures that both $u_l$ and $v_l$ are in the same $\SEEN$ list; hence, their reachability is covered by the parent bag to which they are attached. Thus, the decomposition is valid, and each bag has size at most $\lceil\frac{n}{2}\rceil$. Therefore, $dtw(G_{\text{sub}})\le \lceil\frac{n}{2}\rceil$.
\end{proof}

\begin{remark}
 The above bound is tight. Consider the graph $\gmain = K_4$, and let $G_{\text{sub}}$ be the graph obtained by subdividing each edge of $\gmain$ exactly once. If we orient each edge from the original vertex toward its corresponding subdivision vertex, then we can verify that $dtw(G_{\text{sub}}) = 2$.
\end{remark}
\begin{observation}\label{obs:notminorclosed}
Unlike treewidth $(\mathrm{tw})$, the DAG treewidth $(\mathrm{dtw})$ is not closed under minors. For example, the complete graph $K_6$ has only one source in any acyclic orientation, which implies $\mathrm{dtw}(K_6) = 1$. However, the cycle graph $C_6$, which is a subgraph of $K_6$, satisfies $\mathrm{dtw}(C_6) = 2$. This shows that $\mathrm{dtw}$ is not even closed under subgraphs.
\end{observation}

For a non-source vertex $u$ in a DAG $\vec{H}$, we define the set $P_H(u)$ as the set of source vertices from which $u$ is reachable. That is,
$P_H(u) = \{s \in S \mid u \in \mathrm{R}_H(s)\}$, where $S$ is the set of sources in $\vec{H}$ and $\mathrm{R}_H(s)$ denotes the set of vertices reachable from source $s$. Let $I$ be an induced minor of a graph $\gmain$, then we show the following theorem.

\begin{theorem}\label{dtwinducedminor}
Let $H$ be an undirected graph, and let $I$ be an induced minor of $H$. Then, $dtw(H) \geq dtw(I).$

\end{theorem}

\begin{proof}
To obtain this result, it suffices to prove the following two claims.

    \begin{enumerate}
        \item If $I$ is an induced subgraph of $\gmain$ then $dtw(\gmain)\geq dtw(I)$.  
        \item Let ${H}$ be an undirected graph with $\mathrm{dtw}({H}) = k$. Let $H'$ be the undirected graph obtained by contracting a single edge in ${H}$. Then, $\mathrm{dtw}({H'}) \leq k.$
    \end{enumerate}

\paragraph*{Proof of Claim 1:}
Let $dtw(I) = k$. By definition, there exists an acyclic orientation $O_1$ of $I$ such that $dtw(I) = k$ under this orientation. We now construct an acyclic orientation of $\gmain$ that extends $O_1$ and preserves the treewidth. Let $I' \subseteq \gmain$ be an induced subgraph isomorphic to $I$, with an isomorphism $f: V(I) \to V(I')$. Fix the orientation of $I'$ in $\gmain$ according to $O_1$.

Let $V(\gmain) = V(I') \cup \left(V(\gmain) \setminus V(I')\right)$. Partition the remaining vertices based on their distance from $I'$ in $\gmain$. Define layers $D_1, D_2, \dots, D_j$, where

$$D_i = \{ v \in V(\gmain) \setminus V(I') \mid \text{dist}_{\gmain}(v, V(I')) = i \}.$$

Now, define a topological ordering on the vertices of $\gmain$ as:

$$V(I'), D_1, D_2, \dots, D_j.$$

Orient the edges from each $D_i$ to $D_{i+1}$ in the forward direction (i.e., from smaller to larger layers), and similarly orient the edges from $V(I')$ to $D_1$ forward. For edges within any $D_i$, choose an arbitrary acyclic orientation.
    
Note that after this orientation, we have $S(I)=S(\vec{H})\cap I'$. Also, for any non-source $u$ in $I$, we have $P_{I}(u)=P_{\vec{H}}(u)$ i.e.,  there is no source reaching $u$ other than the source vertices in $I$.

Now, suppose for contradiction that $dtw(\gmain) < k$ under this orientation. Let $T$ be a valid DAG-tree decomposition of $\gmain$ with DAG treewidth less than $k$. Then, the restriction of $T$ to the subgraph $I'$ forms a valid DAG-tree decomposition of $I'$, and hence of $I$ (by isomorphism), contradicting the assumption that $dtw(I) = k$. Therefore, $dtw(\gmain) \geq dtw(I)$.

\paragraph*{Proof of Claim 2:}
The DAG treewidth of an undirected graph $H$ is the maximum of the DAG treewidth of DAG $\vec{H}$. We want to show $dtw(H')\leq k$, where $H'$ is the undirected graph obtained after contraction of an edge $\{u,v\}$. So, for all acyclic orientations of $H'$, we have to show that the DAG treewidth of $H'$ is bounded by $k$. We pick an arbitrary but fixed acyclic orientation of $H'$. We copy the same orientation in $H$. Note that $\{u,v\}\notin E(H')$. So, the orientation is not known. We first give $(u,v)$ orientation to $\{u,v\}$. Note that the orientation of $H$ is acyclic. Let $w$ be the vertex in $H'$ obtained after contracting $\{u,v\}$. Let $T$ be a DAG tree decomposition of $\vec{H}$ of width at most $k$. We will now construct the DAG tree decomposition $T'$ of $\vec{H'}$.

\begin{itemize}
    \item \textbf{Case 1:} Both $u$ and $v$ are non-source vertices.

    In this case, we can assume without loss of generality that $P_H(u) \subseteq P_H(v)$, i.e., all source vertices reaching $u$ also reach $v$. Note that the set of source vertices remains unchanged after contraction, i.e., $S(\vec{H'}) = S(\vec{H})$. We consider the same DAG tree decomposition $T$ of $\vec{H}$. We only need to check whether the source vertices reaching $w$ are connected in $T$. However, one can see that $P_H(v) = P_{H'}(w)$. So, the source vertices in $P_{H'}(w)$ in $T$ are connected.

\item \textbf{Case 2:} $u$ is a source vertex.

We consider two subcases depending on whether the new vertex $w$ is a source in $\vec{H'}$ or not.

\begin{itemize}
    \item  \textbf{Case A :} $w$ is a source vertex.

This happens only when $P_H(v) = \{u\}$, meaning $v$ was reachable only from $u$. In this case, $R_H(u) = R_{H'}(w) \setminus \{v\}$. We can construct a DAG tree decomposition $T'$ of $\vec{H'}$ by simply replacing $u$ with $w$ in all the bags of $T$ where $u$ appears. Since no new vertices are introduced and bag sizes remain unchanged, $\mathrm{dtw}(\vec{H'}) \leq k$.

\item \textbf{Case B:} $w$ is a non-source vertex.

Here, the contraction removes $u$ from the set of sources, so $S(\vec{H'}) = S(\vec{H}) \setminus \{u\}$. $v$ must be reachable by some source vertex $s$ other than $u$. Moreover, $$P_{H'}(w) = P_H(v) \setminus \{u\},$$
and for each source $s \in P_{H'}(w)$, $$R_{H'}(s) = (R_H(s) \cup R_H(u)\cup \{w\}) \setminus \{v\}.$$

For every bag $B$ where $u$ appears and is the only source from $P_H(v)$, we can replace $u$ with another source from $P_H(v)$ present in an adjacent bag. Also, for any $s$ such that $v\in R(s)$, $R(u)\subseteq R(s)$. Thus, we obtain a decomposition of the DAG tree $T'$ with a maximum bag size $k$.

Since these modifications only replace one source with another and do not increase the size of any bag, the resulting DAG tree decomposition $T'$ of $\vec{H'}$ has the same width as $T$, i.e., $\mathrm{dtw}(\vec{H'}) \leq k.$
\end{itemize}
\end{itemize}

\end{proof}

\subsection{Fast Running Time for $\mathrm{hom}(H,G), \mathrm{sub}(H,G) \text{ and } \mathrm{ind}(H,G)$}\label{sec:dtwfast}

In \cite{bera2021near}, the authors established that for any pattern $H$ with at most five vertices, one can count $H$-homomorphisms in bounded-degeneracy graphs in $O(m \log m)$ time. More recently, \cite{paul2024subgraph} presented a subquadratic algorithm for counting subgraphs of patterns with up to nine vertices, achieving a running time of $\tilde{\mathcal{O}}(n^{5/3})$. They further conjectured that no subquadratic algorithm exists for computing the number of subgraphs for all 10-vertex patterns $H$.

Building upon these works (\cite{bera2021near,paul2024subgraph}), we extend the results to patterns with up to 11 vertices and show that they have DAG treewidth 2. Further, we answer the conjecture posed by \cite{paul2024subgraph} in the affirmative.

Further in \cite{bera2021near}, the authors also proved that for a DAG $H$ with $k$ vertices, $dtw(H) \le \frac{k}{4}+2$, and using this obtained upper bounds to compute $\mathrm{hom}(H,G), \mathrm{sub}(H,G) \text{ and } \mathrm{ind}(H,G)$ (\Cref{thm:bressen}). In this section we improve the bound of $\frac{k}{4}+2$ to $\frac{k}{5}+3$ and as a consequence get fast running time for $\mathrm{hom}(H,G), \mathrm{sub}(H,G) \text{ and } \mathrm{ind}(H,G)$ (\Cref{thm:timek/5}).

We now recall the definition of DAG treewidth introduced by Bressan~\cite{bressan2021faster}, along with a key theorem that will help us in deriving the results presented in this section.

\begin{definition} (Similar to Definition 6 of \cite{bressan2021faster})
For an undirected pattern graph $H$, the DAG treewidth of $H$ is $\tau(H) = \tau_3(H)$, where:

\begin{align*}
\tau_1(H) &= \max \left\{ \tau(P) \mid P \in \Sigma(H) \right\} \\
\tau_2(H) &= \max \left\{ \tau_1(H/\theta) \mid \theta \in \Theta(H) \right\} \\
\tau_3(H) &= \max \left\{ \tau_2(H') \mid H' \in D(H) \right\}
\end{align*}

Here:
\begin{itemize}
    \item $\Sigma(H)$ is the set of all acyclic orientations of $H$,
    \item $\Theta(H)$ is the set of all equivalence relations (partitions) over $V_H$, and $H/\theta$ is the pattern obtained by identifying equivalent nodes in $H$ according to $\theta$, removing loops and multiple edges,
    \item $D(H)$ is the set of all supergraphs of $H$ on the same vertex set $V_H$.
\end{itemize}
\end{definition}

We now state the following important theorem from \cite{bressan2021faster}.

\begin{theorem}(Theorem 9 of \cite{bressan2021faster})\label{thm:bressen}
Consider any $k$-node pattern graph $H = (V_H, E_H)$, and let $f_T(k)$ be an upper bound on the time needed to compute a dag tree decomposition of minimum width on $2^{O(k \log k)}$ bags for any $k$-node dag. Then one can compute:
\begin{itemize}
    \item $\mathrm{hom}(H, G)$ in time $2^{O(k \log k)} \cdot O\left(f_T(k) + d^{k - \tau_1(H)} n^{\tau_1(H)} \log n\right)$,
    \item $\mathrm{sub}(H, G)$ in time $2^{O(k \log k)} \cdot O\left(f_T(k) + d^{k - \tau_2(H)} n^{\tau_2(H)} \log n\right)$,
    \item $\mathrm{ind}(H, G)$ in time $2^{O(k^2)} \cdot O\left(f_T(k) + d^{k - \tau_3(H)} n^{\tau_3(H)} \log n\right)$.
\end{itemize}
The theorem still holds if we replace $\tau_1, \tau_2, \tau_3$ with upper bounds, and $f_T(k)$ with the time needed to compute a dag tree decomposition on $2^{O(k \log k)}$ bags that satisfies those upper bounds.
\end{theorem}

As mentioned earlier, \cite{bera2021near} first showed that any pattern $H$ with atmost 5 vertices and later \cite{paul2024subgraph} showed that any pattern $H$ with atmost 9 vertices, subgraph counting can be performed in time $O(m \log m)$ and $\tilde{O}(n^{5/3})$, respectively. We now extend these results to patterns on up to 11 vertices and show that all such patterns have DAG treewidth at most 2.

\begin{theorem}\label{thm:11vertices}
    For any pattern graph $H$, if $|V(H)| \le 11$, then $dtw(\vec{H})\le 2$.
\end{theorem}
\begin{proof}
    If the number of sources is at most 4, we can split them into two bags, each of size at most 2. So, we are interested in cases with at least 5 sources. Also, we assume that $R(s)$ is not contained in $R(s_i)$ for any source $s_i$, otherwise we can attach a leaf bag containing $s$ to a parent bag containing $s_i$. Further, we can assume that $|R(s) | \geq 2$, otherwise it is easy to obtain a DAG tree decomposition without increasing the bag size.

    \begin{itemize}
        \item \textbf{Case 1:} Assume that the number of sources is 5 and the number of non-sources is 6.

        We pick a source $s$ such that $|R(s)|$ is minimum.

        \begin{itemize}
            \item \textbf{Case A:} If $|R(s)|\le 2$ then we can take two sources $s_i,s_j$ such that $R(s)\subseteq R(s_i)\cup R(s_j)$. Thus,  we can put $s$ in leaf bag $B$ and put $B_1=\{s_i,s_j\}$ adjacent to $B$. We place the remaining two source vertices in $B_2$, adjacent to $B_1$. It is easy to verify the reachability condition, as we only have to argue about the reachability intersection of $s$ and $R(B_2)$. However, $R(s)\subseteq R(B_1)$. Similarly, we can argue that whenever there exist sources $s$ such that there exists a pair $s_i, s_j$, then $R(s)\cap R(s_l)\subseteq R(s_i)\cup R(s_j)$ for any source $s_l$. Then we can get the DAG tree decomposition of the bag size at most two.

            Now consider the case when $|R(s)|\ge 3$ and there does not exist $\{s_i,s_j\}$ such that $R(s)\subseteq R(s_i)\cup R(s_j)$.

            \item \textbf{Case B:} $|R(s)|=3$

            We know that for any $s_i$, $|R(s_i) \cap R(S)|\le 1$, otherwise $R(s) \subseteq R(s_i)\cup R(s_j)$. Therefore, $|R(s_i)\cup R(s)| \ge 5$. If $|R(s_i)\cup R(s)| =6$ then we are done.

            If there is no such $s_i$, then we know that for each source $s$,$|R(s)|=3$. So, we consider the case when $R(s)=3$ and $|R(s_i)\cap R(s_j)|\le 1$ for every pair of source vertices $s_i$ and $s_j$. However, we require at least seven non-source vertices even to satisfy the condition when the number of sources is four. Thus, it is not possible.
            
            \item \textbf{Case C:} $|R(s)|=4$

            We know that for any $s_i$, $|R(s_i)\cap R(s)|\le 2$. But then $|R(s_i)\cup R(s)|\ge 6$. So we are done.
        \end{itemize}

        \item \textbf{Case 2:} Assume that the number of sources is 6 and the number of non-sources is 5.

        \begin{itemize}
            \item \textbf{Case A:} If there exists $s_i$ such that $R(s_i)\geq 3$.

            If $\exists s_j$ such that $|R(s_i)\cup R(s_j)|=5$, then we are done. So, $|R(s)\setminus R(s_i)|\leq 1$.

            For a non-source $x$ such that $x\notin R(s_i)$, and $|P(x)|$ is at least three (number of sources reaching $x$), there exists $x$ because there are only two non-sources that are not reachable by $s$, and there are a total $5$ source vertices other than $s$. We pick a source $s$, s.t. $x \in R(s)$. We make bag $B_1=\{s_i,s\}$. The source vertices that reach $x$ can be kept as leaf bags in $B_1$. The source vertices that do not reach $x$ are at most two. We keep it in a bag $B_2$. This gives a valid DAG tree decomposition.

            \item \textbf{Case B:} $|R(s_i)|\le 2$

            If $|R(s_i)|\le 2$ then pick two sources $s_1$ and $s_2$ such that $|R(s_1)\cap R(s_2)|=1$. Put $s_1$ and $s_2$ in bag $B_1$. Now take two other sources $s_3$ and $s_4$ such that $|R(s_1)\cap R(s_3)|=1$ and $|R(s_2)\cap R(s_4)|=1$ and put $s_3$ and $s_4$ in bag $B_2$. Now, for the remaining sources $s_5$ and $s_6$, put them in bag $B_3$. Now we make a DAG tree decomposition where the edges are between $ B_1$, $ B_2$, and $ B_3$.
        \end{itemize}

        \item \textbf{Case 3:} When the number of sources is 7, then the number of non-sources is 4.

        If $\exists s$ such that $|R(s)|\ge 3$ then we are done. Therefore $|R(s)|\le 2$. Now, if $|R(s)|=1$, then it can be easily handled. So we only consider the case when $|R(s)|=2$. However, in this case, it is not possible to have 7 sources such that $R(s_i) \nsubseteq R(s_j)$ for any $ i \ ne j$. Therefore, we have $|R(s_i)\cup R(s_j)|=4$. Hence, we are done. 
    \end{itemize}

    The remaining cases are also similar to the proof.
\end{proof}
 By applying ~\Cref{thm:11vertices} together with ~\Cref{thm:bressen}, we derive the following improved bounds.

\begin{theorem}\label{thm:11vertex}
  For any pattern graph $H$ up to 11 vertices, we can compute:
\begin{itemize}
    \item $\mathrm{hom}(H, G)$ in time $ O\left(d^{k -2} n^{2} \log n\right)$,
    \item $\mathrm{sub}(H, G)$ in time $ O\left(d^{k -2} n^{2} \log n\right)$,
    \item $\mathrm{ind}(H, G)$ in time $ O\left(d^{k -2} n^{2} \log n\right)$.
\end{itemize}
\end{theorem}
\begin{proof}
    From \Cref{thm:11vertices}, we know that all the patterns with eleven vertices have $dtw$ at most two.
    
    Given an undirected graph $G$, we know that the number of vertices of $\spasm(G)$ is at most the number of vertices in $G$. Therefore $\tau_2(H)\leq 2$. Similarly, the number of vertices in $G$ is bounded for every supergraph of $\spasm(G)$. Therefore, $\tau_3(H)\leq 2$. Therefore, using \Cref{thm:bressen} we get the desired runtime for $\mathrm{hom}(H,G), \mathrm{sub}(H,G) \text{ and } \mathrm{ind}(H,G)$.
\end{proof}

The above bound on eleven vertices is tight for quadratic time, as there exists a pattern on twelve vertices that has $dtw(H)=3$. Consider the pattern in \Cref{fig:X}. We show that $dtw(X)=3$.

\begin{lemma}
Let \(\vec{X}\) be the DAG on twelve vertices shown in \Cref{fig:X}. Then \(\operatorname{dtw}(X)=3\).
\end{lemma}

\begin{proof}
We consider the acyclic orientation of \(X\) in which the black vertices in \Cref{fig:X} are the source vertices. Since there are six sources, we can construct a DAG-tree decomposition in which each bag contains at most three sources by distributing the sources across two bags. This yields a DAG-tree decomposition of width three, and hence \(\operatorname{dtw}(X)\le 3\).

To prove the lower bound, consider the three source vertices \(\{s_1,s_2,s_3\}\). Any DAG-tree decomposition must either contain a bag that includes all three of these sources, or contain two adjacent bags \(B_1\) and \(B_2\) such that two of the sources, say \(s_i\) and \(s_j\), appear in \(B_1\) and the remaining source \(s_k\) appears in \(B_2\). 

Now consider the remaining source vertices \(u_1, u_2, u_3\). Observe that for each \(i\), the reachability set \(R(u_i)\) is contained in \(R(s_i)\cup R(s_j)\) for suitable \(i\). However, the vertices \(u_j\) and \(u_k\) cannot be placed in non-adjacent bags without violating the reachability intersection property. Consequently, there must exist a bag containing the three sources \(\{s_k,u_j,u_k\}\). This implies that every DAG-tree decomposition of \(X\) has a bag of size at least three, and therefore \(\operatorname{dtw}(X)\ge 3\).

Combining both bounds, we conclude that \(\operatorname{dtw}(X)=3\).
\end{proof}

\begin{figure}[t]
\centering
\begin{tikzpicture}[
    vertex/.style={circle, draw, thick, inner sep=0.8mm},
    source/.style={vertex, fill=black},
    nonsource/.style={vertex, fill=blue}
]

\node[source, label=left:$u_1$]  (u1) at (0,4) {};
\node[source, label=right:$u_3$] (u3) at (4,4) {};
\node[source, label=below:$u_2$] (u2) at (2,1) {};

\node[source, label=right:$s_1$] (s1) at (2,4) {};
\node[source, label=left:$s_2$]  (s2) at (1,2.5) {};
\node[source, label=right:$s_3$] (s3) at (3,2.5) {};

\node[nonsource, label=above:$v_1v_6$] (v16) at (2,5) {};
\node[nonsource, label=below:$v_2v_3$] (v23) at (0,2) {};
\node[nonsource, label=right:$v_4v_5$] (v45) at (4,2) {};

\node[vertex, label=left:$u'_1$]  (u1p) at (1,3.5) {};
\node[vertex, label=below:$u'_2$] (u2p) at (2,2.5) {};
\node[vertex, label=above:$u'_3$] (u3p) at (3,3.5) {};

\draw (u1)--(v16)--(u3)--(v45)--(u2)--(v23)--(u1);
\draw (v16)--(s1)--(u3p)--(s3)--(v45);
\draw (s1)--(u1p)--(s2)--(v23);
\draw (s3)--(u2p)--(s2);

\end{tikzpicture}
\caption{The graph \(X\). The black vertices denote sources, and the others are non-sources.}
\label{fig:X}
\end{figure}

In \cite{paul2024subgraph}, apart from the subquadratic algorithm for subgraph counting of patterns with at most 9 vertices, the authors conjectured the following:

\conjjj*

In the above conjecture, $\mathcal{H}_{\Delta}$ is the same as defined in \cite{paul2024subgraph} (See \Cref{fig:4vertices}). Thus, \Cref{con:paul} implies that there are no subquadratic algorithms for computing the number of subgraphs of all 10-vertex patterns $H$. We now answer this conjecture in the affirmative by showing that no combinatorial algorithm can compute $\mathrm{sub}(H, G)$ in $o(m^2)$ time, when $|V(H)|=10$.

\subsection*{Proof of \Cref{con:paul}}

Our proof relies on the assumption from~\cite{komarath2023finding}, which states that there is no combinatorial algorithm of time $o(m^2)$ to count copies of $K_4$ ( \cite{komarath2023finding}).

We use a strategy similar to that in~\cite{bera2021near} to prove the hardness of counting subdivision of $K_4$ patterns.  

Let \( H \) be a graph on \( h \) vertices. We define a coloring of a graph \( G \) as 
\( c: V(G) \rightarrow [h] \). A mapping \( \phi: V(H) \rightarrow V(G) \) is called a 
\textit{colorful homomorphism} if 
\[
\{ c(\phi(v)) : v \in V(H) \} = [h].
\]
In other words, a colorful homomorphism maps each vertex of \( H \) to a distinct color class in \( G \), thereby ensuring injectivity. It is a known fact that counting injective colorful homomorphisms reduces to counting homomorphisms \( \mathrm{Hom}(H, G) \) such that the 
reduction preserves degeneracy (see Lemma 3.1 \cite{bera2021near}). Also, similar to Lemma 3.2 of \cite{bera2021near}, if one can detect colorful \( K_4 \) in a 4 vertex colored graph with \( m \) edges in time \( f(m) \), then it is possible to detect (uncolored) \( K_4 \) in a graph with \( m \) edges in time \( \tilde{O}(f(m)) \) via color-coding.

Using all these known facts, we now describe the construction used for this reduction: 
Given a graph \( G \) and a coloring \( c: V(G) \rightarrow [4] \), we construct a new graph 
\( G' \) with a coloring \( c': V(G') \rightarrow [10] \) in time \( O(|V(G)| + |E(G)|) \), 
such that 
\[
|V(G')|,\, |E(G')| = \mathcal{O}(|V(G)| + |E(G)|).
\]
We now claim that the number of colored subdivision \( K_4 \) patterns in \( G' \) is equal to the 
number of \( K_4 \) subgraphs in \( G \). Furthermore, \( G' \) is 2-degenerate.

The construction is as follows: Since \( G \) is colored with 4 colors, we assign one of 
6 new colors to each subdivision vertex in \(G' \). For every edge \( \{x, y\} \in E(G) \), 
if \( c(x) \neq c(y) \), we subdivide the edge by introducing a new vertex with color 
\( (c(x), c(y)) \). If \( c(x) = c(y) \), the edge is deleted. Note that the original vertices 
in \( V(G) \) retain colors from \(\{1, 2, 3, 4\}\), while subdivision vertices are colored 
with the remaining 6 colors. Clearly, $|V(G')| \leq |V(G)| + |E(G)| \quad \text{and} \quad |E(G')| \leq 2|E(G)|$
both remain in \( O(|V(G)| + |E(G)|) \).

One can observe that the original vertices \( V(G) \) form an independent set in \(G' \), and 
all subdivision vertices have degree exactly 2 in \(G' \). After removing the subdivision 
vertices, the resulting graph is empty. Hence, \( G' \) is 2-degenerate. Also, from the construction, it is straightforward to verify that the number of colorful \( K_4 \) copies in \( G \) equals the number of colorful subdivision \( K_4 \) copies in \(G' \).

Now, suppose that homomorphism counting for subdivision \( K_4 \) can be performed in 
$o((n^2)$ time on 2-degenerate \( n \)-vertex graphs. 
Then, one can also count colored injective homomorphisms of  subdivision \( K_4 \) in 
$o((n^2)$ time on 2-degenerate, \( k \)-vertex-colored, \( n \)-vertex graphs.  

Using the above construction, one can then detect colorful \( K_4 \) in 4-vertex-colored, 
\( m \)-edge graphs in $o(m^2)$ time. Consequently, this yields an 
\( \tilde{o}(m^2) \)-time algorithm for detecting (uncolored) \( K_4 \) in general 
\( m \)-edge graphs.

However, it is known that \( \mathrm{mtw}(K_4) = 4 \) and no combinatorial algorithm exists 
to detect \( K_4 \) in $o(m^2)$ time under standard complexity assumptions (\cite{komarath2023finding}). 
Therefore, subdivision \( K_4 \) subgraph counting cannot be done in \( o(m^2) \) time. This further implies that there is no $o(m^2)$ algorithm for computing $sub(H, G)$, when $V(H)=10$.

\hfill $\square$
\begin{figure}[h]
    \centering
    \captionsetup{justification=centering}
    \includegraphics[width=0.5\textwidth]{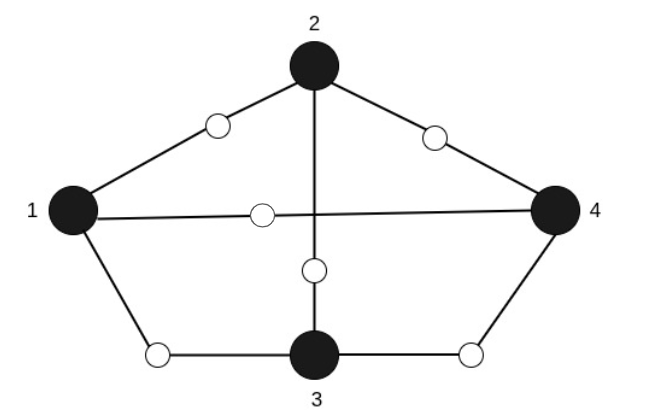}
    \caption{$k_4$ with one subdivision}
    \label{fig:4vertices}
\end{figure}

We now derive the following lemma, which will help us to show that for a DAG with $k$ vertices, $dtw(H)\le \frac{k}{5}+4$. We then do further refinement to show that $dtw(H)\le \frac{k}{5}+3$.

\begin{lemma}\label{lem:boundedreachablity}
    Given a DAG $H$ with $t$ source vertices such that the reachability of every source vertex is bounded by 2, we have $dtw(H) \leq \frac{t}{5.217} + 4$.
\end{lemma}

\begin{proof}
    We construct a graph $G$ where each vertex corresponds to a non-source vertex in $H$. For every source vertex $s$ in $H$ such that $R(s) = \{u, v\}$, we add the edge $\{u, v\}$ in $G$. Thus, the number of edges in $G$ is equal to the number of source vertices in $H$, i.e., $|E(G)| = t$.

    From the result of ~\cite{kneis2005algorithms}, we know that any graph with $m$ edges has treewidth at most $\frac{m}{5.217} + 3$. Hence, $tw(G) \leq \frac{t}{5.217} + 3$, and there exists a tree decomposition $T$ of $G$ with maximum bag size at most $\frac{t}{5.217} + 4$.

    We now use $T$ to construct a DAG tree decomposition $T'$ of $H$. For every non-source vertex $u$ in a bag of $T$, we replace it by the source vertex $s$ such that $u \in R(s)$. Since $|R(s)| = 2$ for every $s$, this replacement is well-defined and unique. After this replacement, for every source vertex $s$, either $s$ appears in some bag, or $R(s)$ is contained in a single bag. In the latter case, we can add a leaf bag $B_s$ containing only $s$ as a child of the bag that contains $R(s)$. This ensures the reachability condition of DAG tree decomposition is satisfied.

    Thus, the resulting DAG tree decomposition $T'$ has maximum bag size at most $\frac{t}{5.217} + 4$, and we conclude that $dtw(H) \leq \frac{t}{5.217} + 4$.
\end{proof}

\begin{theorem}\label{thm:kbyfive}
    For any DAG $H$ with $k$ vertices, we have $dtw(H) \leq \frac{k}{5} + 4$.
\end{theorem}

\begin{proof}
    We construct a DAG tree decomposition for $H$ with maximum bag size at most $\frac{k}{5} + 4$.

    The construction has two phases. In the first phase, we build a set of source vertices through which we can discover at least four new non-sources. In the second phase, we build a DAG decomposition for the remaining sources. Finally, we combine the two to obtain a DAG tree decomposition of $H$.

    \paragraph*{First phase:} Initialize $B^*$, a subset of sources of $H$, to the empty set. We greedily add sources to $B^*$ as follows: at each step, if there is a source vertex $s$ such that $|R(s) \setminus R(B^*)| \geq 4$. We mark $s$ and the non-sources reachable from $s$ as $\SEEN$. This step marks at least five new vertices. We add $s$ to $B^*$ and remove all $\SEEN$ vertices from $H$. Then, we repeat this process. If there is no source vertex $s$ such that $|R(s) \setminus R(B^*)| \geq 4$, then we stop this phase. Let the subgraph of $H$ at the end of this phase be $H'$, and let $k'$ be the number of vertices in $H'$.
    
    \paragraph*{Second phase:} Our goal now is to construct a DAG tree decomposition $T'$ of $H'$ with maximum bag size at most $\frac{k'}{5} + 4$. In the graph $H'$, every source vertex has at most three reachable vertices, that is, $|R(s)| \leq 3$. We now apply the construction in the first phase again to build a set of sources $B_1$ (starting from the empty set) by adding iteratively sources $s$ such that $|R(s) \setminus R(B_1)| = 3$ while removing all vertices reachable from $s$ at each step as before. Let $H_1$ be the resulting subgraph of $H'$. Now, we split the proof into cases:
\begin{itemize}
    \item $|B_1| = 0$: By ~\Cref{lem:boundedreachablity}, we have $dtw(H') \leq \frac{k'}{5.217} + 4$.

    \item $|B_1| = \frac{k'}{5}$: We have removed at least $4(\frac{k'}{5})$ vertices while building $B_1$. Therefore, at most $\frac{k'}{5}$ sources remain. Let $B_2$ be the set of remaining sources. The DAG tree decomposition of $H'$ is a tree that is an edge. One vertex is the bag $B_1$, and the other is the bag $B_2$. This decomposition has a maximum bag size $\frac{k'}{5}$.

    \item $0 < |B_1| < \frac{k'}{5}$: Let $t$ be the number of non-source vertices in $H_1$. If $t \leq \frac{k'}{5}$, we build a bag $B'$ that includes one source for each of the $t$ non-source vertices. Any remaining sources are attached to the bag $B'$ as singleton leaf bags. Finally, the bag $B_1$ is attached to $B'$. The maximum bag size remains at most $\frac{k'}{5}$. If $t > \frac{k'}{5}$, then the number of sources in $H_1$ is less than $\frac{4k'}{5}$. Let $G_1$ be the induced subgraph of $H'$ on $B_1 \cup R(B_1)$ and let $k_1$ be the number of vertices in $G_1$. Then, $G_1$ has $\frac{k_1}{4}$ sources and $\frac{3k_1}{4}$ non-sources. Let $G_2$ be the graph obtained by removing $V(G_1)$ from $H'$. In $G_2$, each source has at most two reachable vertices, and the number of sources is less than $\frac{4k'}{5} - \frac{k_1}{4}$. Applying ~\Cref{lem:boundedreachablity} gives:
    $$
    dtw(G_2) \leq \frac{1}{5.217} \left( \frac{4k'}{5} - \frac{k_1}{4} \right) + 4.
    $$
    We then form the decomposition of $H'$ by adding $B_1$ to each bag in the decomposition of $G_2$. So, the maximum bag size is:
    $$
    dtw(H') \leq \frac{4k'}{25} - \frac{k_1}{20} + \frac{k_1}{4} + 4 = \frac{4k'}{25} + \frac{k_1}{5} + 4.$$
    If $k_1 > \frac{k'}{5}$, then the number of source vertices in $H'$ is at most $\frac{2k'}{5}$. Thus, we can get DAG tree decomposition for $H'$. with maximum bag size $\frac{k'}{5}$. So $k_1 \leq \frac{k'}{5}$ and therefore $\frac{4k'}{25} + \frac{k_1}{5} + 4 \leq \frac{k'}{5} + 4$.
\end{itemize}
Therefore, in all cases, at the end of the second phase, we have a DAG tree decomposition for $H'$ with maximum bag size at most $k'/5 + 4$.

Finally, by adding $B^*$ from the first phase to every bag of the decomposition of $H'$ obtained in the second phase, we obtain a DAG tree decomposition of $H$. The maximum bag size of this DAG tree decomposition is: $|B^*|+\frac{k'}{5}+4 \leq \frac{(k-k')}{5}+\frac{k'}{5}+4=\frac{k}{5}+4$.
\end{proof}

\begin{corollary}\label{cor:k/4+3}
    The additive term $+4$ arises only when using ~\Cref{lem:boundedreachablity}, which divides by $5.217$. If we divide by $5$, then in practice, using $+3$ is sufficient. Hence, $dtw(H) \leq \frac{k}{5} + 3$ for all DAGs with $k$ vertices.
\end{corollary}

We now use \Cref{cor:k/4+3} and \Cref{thm:bressen} to have the following result.

\begin{theorem}\label{thm:timek/5}
  Consider any $k$-node pattern graph $H = (V_H, E_H)$. Then one can compute:
\begin{itemize}
    \item $\mathrm{hom}(H, G)$ in time $ 2^{O(k\log k)}\cdot O(d^{k-\lfloor\frac{k}{5}\rfloor-3}n^{\lfloor\frac{k}{5}\rfloor+3}\log n)$,
    \item $\mathrm{sub}(H, G)$ in time  $ 2^{O(k \log k)}\cdot O(d^{k-\lfloor\frac{k}{5}\rfloor-3}n^{\lfloor\frac{k}{5}\rfloor+3}\log n)$,
    \item $\mathrm{ind}(H, G)$ in time $ 2^{O(k^2)}\cdot O(d^{k-\lfloor\frac{k}{5}\rfloor-3}n^{\lfloor\frac{k}{5}\rfloor+3}\log n)$.
\end{itemize}
\end{theorem}

This improves the previous known bound for counting induced subgraphs of a $k$-vertex pattern of \cite{bressan2021faster}. In \cite{bressan2021faster}, the author improved the exponent of $n$ from  $0.791k + 2$ (\cite{nevsetvril1985complexity}) to $0.25k+2$. In  \Cref{cor:k/4+3}, we show that the exponent of $n$ is bounded by $0.2k+3$. Let $r$ be the average degree of $G$. Since we know that $d=O(\sqrt{rn})$ as a corollary of \Cref{thm:timek/5} we have the following theorem:

\begin{theorem}\label{thm:averagedegree}
    Consider any $k$-node pattern graph $H=(V_H,E_H)$ and let $r$ be the average degree of $G$. Then one can compute:
\begin{itemize}
    \item $\mathrm{hom}(H, G)$ in time $ 2^{O(k\log k)}\cdot O(r^{\frac{1}{2}(k-\lfloor\frac{k}{4}\rfloor-2)}n^{\frac{1}{2}(k+\lfloor\frac{k}{4}\rfloor+2)}\log n)$,
    \item $\mathrm{sub}(H, G)$ in time  $ 2^{O(k \log k)}\cdot O(r^{\frac{1}{2}(k-\lfloor\frac{k}{4}\rfloor-2)}n^{\frac{1}{2}(k+\lfloor\frac{k}{4}\rfloor+2)}\log n)$,
    \item $\mathrm{ind}(H, G)$ in time $ 2^{O(k^2)}\cdot O(r^{\frac{1}{2}(k-\lfloor\frac{k}{4}\rfloor-2)}n^{\frac{1}{2}(k+\lfloor\frac{k}{4}\rfloor+2)}\log n)$.
\end{itemize}
\end{theorem}

\section{Connections to Treewidth and Treedepth}\label{sec:connection}

As DAG treedepth is a new parameter introduced in this work, inspired by the classical notion of treedepth, our goal is to establish a precise relationship between the two. Using the construction $\mathcal{G}$ defined in ~\Cref{sec:construction}, we show that for $G_1, G_2 \in \mathcal{G}$, the following holds: $td(G_1) \leq dtd(\vec{H}) \leq td(G_2).$

\begin{lemma}\label{lem:dtdupper}
For a given DAG $\vec{H}$ there exists a $G_S\in \mathcal{G}$ such that $dtd(\vec{H}) \leq td(G_S)$.
\end{lemma}

\begin{proof}
Let $T$ be an elimination tree of the graph $G_S$. From the construction in ~\Cref{sec:construction}, the vertex set of $G_S$ is $V(G_S) = S$, where $S$ is the set of sources in the DAG $\vec{H}$. By definition, every source vertex appears in $T$.

To prove the lemma, we need to show that $T$ also satisfies the \emph{reachability intersection property} required for a valid DAG elimination tree. That is, for any pair of source vertices $s_i$ and $s_j$, if their reachability sets intersect ($\mathrm{R}_H(s_i) \cap \mathrm{R}_H(s_j) \neq \emptyset$), then either $s_i$ and $s_j$ lie on the same root-to-leaf path in $T$, or every vertex in $\mathrm{R}_H(s_i) \cap \mathrm{R}_H(s_j)$ is reachable from some common ancestor source $s_k$ of both $s_i$ and $s_j$ in $T$.

Formally, we require that $\mathrm{R}_H(s_i) \cap \mathrm{R}_H(s_j) \subseteq \bigcup_k \mathrm{R}_H(s_k),$ where the union is over all source vertices $s_k$ that are common ancestors of $s_i$ and $s_j$ in $T$. If $\mathrm{R}_H(s_i) \cap \mathrm{R}_H(s_j) = \emptyset$, there is nothing to prove. If $s_i$ and $s_j$ lie on the same root-to-leaf path in $T$, the reachability intersection property is also satisfied.

Now consider the case when $\mathrm{R}_H(s_i) \cap \mathrm{R}_H(s_j) \neq \emptyset$ and $s_i$, $s_j$ lie on different root-to-leaf paths in $T$. This can happen only if $\{s_i, s_j\} \notin E(G_S)$. From the construction of $G_S$, we know that for each non-source vertex $u \in \mathrm{R}_H(s_i) \cap \mathrm{R}_H(s_j)$, there exists a source $s_u$ such that $u \in \mathrm{R}_H(s_u)$, and during the edge contraction process, edges $\{s_i, s_u\}$ and $\{s_j, s_u\}$ were added to $G_S$. Therefore, $s_u$ is a common ancestor of $s_i$ and $s_j$ in $T$, and $u$ is in the reachability set of $s_u$.

Since $u$ was chosen arbitrarily from $\mathrm{R}_H(s_i) \cap \mathrm{R}_H(s_j)$, the entire intersection is covered by the reachability sets of ancestor sources in $T$. Hence, $T$ satisfies the reachability intersection condition. Therefore, $T$ is a valid DAG elimination tree of $\vec{H}$, and we conclude that $dtd(\vec{H}) \leq td(G_S)$.
\end{proof}

\begin{lemma}\label{lem:dtdlb}
   For a given DAG $\vec{H}$ there exist a $G_S\in \mathcal{G}$ such that,  $dtd(\vec{H}) \geq td(G_S)$.
    
\end{lemma}

\begin{proof}
    Let $T$ be a DAG elimination tree of $\vec{H}$ based on a fixed acyclic orientation. Recall from the construction in ~\Cref{sec:construction} that the vertex set of $G_S$ is the set of source vertices of $\vec{H}$. Therefore, the nodes of $T$ correspond to the sources in $H$. We aim to show that $T$ is also a valid elimination tree for $G_S$, thereby proving that the treedepth of $G_S$ is at most the DAG treedepth of $\vec{H}$.

    Consider any pair of source vertices $s_i$ and $s_j$:
    \begin{itemize}
        \item If $s_i$ and $s_j$ lie on the same root-to-leaf path in $T$, then the elimination tree requirement for $G_S$ is trivially satisfied.
        \item If $R(s_i) \cap R(s_j) = \emptyset$, then $s_i$ and $s_j$ are not adjacent in $G_S$, and again there is nothing to prove.
        \item Otherwise, suppose $R(s_i) \cap R(s_j) \neq \emptyset$ and $s_i$ and $s_j$ lie on different root-to-leaf paths in $T$. Since $T$ is a valid DAG elimination tree, it satisfies the reachability intersection condition. For each non-source vertex $u \in R(s_i) \cap R(s_j)$, there exists a source vertex $s_u$ that is a common ancestor of both $s_i$ and $s_j$ in $T$ such that $u \in R(s_u)$.

         By the construction of $G_S$, for such a vertex $u$, an edge contraction was performed between $u$ and $s_u$. Therefore, no edge exists between $s_i$ and $s_j$ in $G_S$, the shared reachability is captured by their ancestor $s_u$.
    \end{itemize}

    Hence, $T$ is a valid elimination tree for $G_S$, and its depth is at least the treedepth of $G_S$. Thus,$dtd(\vec{H}) \geq td(G_S).$
\end{proof}

Next, as the relationship between treewidth and DAG treewidth remains relatively unexplored, we aim to establish a more concrete connection between them. Leveraging the construction $\mathcal{G}$ defined in ~\Cref{sec:construction}, we prove that there exists graphs $G_1, G_2 \in \mathcal{G}$ such that the following bounds hold: $\frac{\mathrm{dtw}(G_1) + 1}{2} \leq \mathrm{dtw}(\vec{H}) \leq \mathrm{dtw}(G_2) + 1.$ To this end in \Cref{lem:dtwupper} we show that $dtw{\vec{H} \le tw(G_S)}+1$ and in \Cref{lem:dtwlb} we show that $dtw(\vec{H}) \geq \frac{tw(G_S)+1}{2}.$

\begin{lemma}\label{lem:dtwupper}
For a given DAG $\vec{H}$ there exist a $G_S\in \mathcal{G}$ such that $dtw(\vec{H})\leq tw(G_S)+1$.
\end{lemma}
\begin{proof}
$T$ be a tree decomposition of the graph $G_S$. From the construction described in \Cref{sec:construction}, we know that $V(G_S) = S$, where $S$ is the set of sources in the DAG $\vec{H}$. By the definition of a tree decomposition, each bag in $T$ is a subset of the vertex set $V(G_S) = S$, and the union of all bags covers $S$. Thus, $T$ satisfies the \emph{coverage property} of the DAG-tree decomposition (DTD) (see \Cref{def:dagtreedecomposition}).

To complete the proof of the lemma, it remains to show that $T$ also satisfies the \emph{reachability intersection property}. Specifically, for all bags $B, B_1, B_2 \in \mathcal{B}$, if $B$ lies on the unique path between $B_1$ and $B_2$ in the tree $T$, then we must have:

$$\mathrm{R}_H(B_1) \cap \mathrm{R}_H(B_2) \subseteq \mathrm{R}_H(B).$$

To prove this, we define: $Y_1 = B_1 \setminus B_2$ and $Y_2 = B_2 \setminus B_1$. Since $T$ is a tree decomposition, the \emph{connectivity property} implies that for any vertex $v \in B_1 \cap B_2$, all bags containing $v$ must lie along the path between $B_1$ and $B_2$, which includes $B$. Therefore, $B_1 \cap B_2 \subseteq B$, and consequently:
$$\mathrm{R}_H(B_1 \cap B_2) \subseteq \mathrm{R}_H(B).$$

Thus, to prove the reachability intersection condition, it is sufficient to consider the contributions from $Y_1$ and $Y_2$. This leads us to two cases:

\begin{enumerate}
    \item  $\mathrm{R}_H(Y_1) \cap \mathrm{R}_H(Y_2) = \emptyset$.

    \item $\mathrm{R}_H(Y_1) \cap \mathrm{R}_H(Y_2) \neq \emptyset$.
\end{enumerate}

For both cases, we want to show that the reachability condition holds.

\

\paragraph*{Case 1: $\mathrm{R}_H(Y_1) \cap \mathrm{R}_H(Y_2) = \emptyset$}

\hfill \break

We know that $\mathrm{R}_H(B_1) = \mathrm{R}_H(Y_1 \cup (B_1 \cap B_2)),$ and $\mathrm{R}_H(B_2) = \mathrm{R}_H(Y_2 \cup (B_1 \cap B_2)).$

Taking the intersection we have,
$$\mathrm{R}_H(B_1) \cap \mathrm{R}_H(B_2) = \left( \mathrm{R}_H(Y_1) \cap \mathrm{R}_H(Y_2) \right) \cup \mathrm{R}_H(B_1 \cap B_2).$$

Since $\mathrm{R}_H(Y_1) \cap \mathrm{R}_H(Y_2) = \emptyset$ by assumption, the intersection simplifies to:

$$\mathrm{R}_H(B_1) \cap \mathrm{R}_H(B_2) = \mathrm{R}_H(B_1 \cap B_2).$$

From the connectivity property of tree decompositions, $B_1 \cap B_2 \subseteq B$, so:$$\mathrm{R}_H(B_1 \cap B_2) \subseteq \mathrm{R}_H(B).$$

Thus, the reachability condition holds in this case.

\paragraph*{Case 2: $\mathrm{R}_H(Y_1) \cap \mathrm{R}_H(Y_2) \neq \emptyset$} 

This is the more involved case where the reachability sets of $Y_1$ and $Y_2$ overlap. Our goal remains to show that every vertex in $\mathrm{R}_H(B_1) \cap \mathrm{R}_H(B_2)$ is also in $\mathrm{R}_H(B)$.

Suppose there exist sources $u \in Y_1$ and $v \in Y_2$ such that $\mathrm{R}_H(u) \cap \mathrm{R}_H(v) \neq \emptyset$. We consider two subcases based on how the common reachable non-source vertices are connected in the graph $G_S$.

\begin{enumerate}[label=(\Alph*)]
    \item \textit{There exists a non-source vertex that is reachable only from $u$ and $v$:}

From the construction of the graph $G_S$, we know that in such a case, the edges connecting this non-source vertex to $u$ and $v$ would be contracted. This results in an edge $\{u, v\} \in E(G_S)$, implying that $T$, the tree decomposition of $G_S$, contains a bag with both $u$ and $v$.

\vspace{3mm}

\item \textit{The common reachable non-source vertices are also reachable from some other set of sources:}

If, for any such non-source vertex in $\mathrm{R}_H(u) \cap \mathrm{R}_H(v)$, the edge contractions in $G_S$ were performed using $u$ and $v$, then again $\{u, v\} \in E(G_S)$, and $T$ contains a bag with both $u$ and $v$.

However, consider the case where edge contractions were performed using a different set of sources, say, a set $\{S' \}$ such that each non-source vertex in $\mathrm{R}_H(u) \cap \mathrm{R}_H(v)$ is reachable from some source in  $S'$. These sources in $S'$ are then connected via contractions to the non-sources and become common neighbors of both $u$ and $v$ in $G_S$. As a result, the tree decomposition $T$ must include intermediate bags containing these shared neighbors, which lie on the path between $u$ and $v$ in $T$.

Since $u$ and $v$ were chosen arbitrarily from $Y_1$ and $Y_2$, this argument extends to show that all shared reachable vertices in $\mathrm{R}_H(Y_1) \cap \mathrm{R}_H(Y_2)$ must be included in the reachability of some intermediate bag $B$ on the path between $B_1$ and $B_2$. Therefore,

$$
\mathrm{R}_H(B_1) \cap \mathrm{R}_H(B_2) \subseteq \mathrm{R}_H(B),
$$

for all such bags $B$ on the path from $B_1$ to $B_2$ in $T$.
\end{enumerate}

 This shows that $T$ satisfies both the coverage and the reachability intersection properties required for a valid DAG-tree decomposition. Finally, by definition, the DAG-treewidth of $\Vec{H}$ is the maximum size of any bag in the DAG-tree decomposition. Since the treewidth of $G_S$ is one less than the size of its largest bag, we conclude that $\mathrm{dtw}(\Vec{H}) \leq \mathrm{tw}(G_S) + 1.$

\end{proof}

Using  \Cref{obs:equaltreewidth} and \Cref{lem:dtwupper} we get the following remark:
\begin{remark}
    For a given DAG $\Vec{H}$ and corresponding $\textsc{Bip}(\vec{H})$, $dtw(H)\leq tw(\textsc{Bip}(\vec{H}))+1$.
\end{remark}

\begin{remark}
    
The above bound is tight. For example, consider the cycle graph $C_5$. For any acyclic orientation of $C_5$, the DAG treewidth is $\mathrm{dtw}(C_5) = 1$, while the treewidth is $\mathrm{tw}(C_5) = 2$.

For a more general case, consider the graph $G$ obtained by subdividing every edge of the complete bipartite graph $K_{n,n}$ exactly once. Use an orientation where the original vertices have out-degree only (i.e., all original vertices are sources), and the newly added subdivision vertices are non-sources.

In this construction, each vertex in one part of the bipartite graph has a unique reachable non-source shared with every vertex in the opposite part. By  ~\Cref{obs:uniqueintersection}, all sources from one side must appear together (or in adjacent bags) in any DAG tree decomposition. Hence, any bag must contain at least $n$ source vertices among the total of $2n$ sources. This gives $\mathrm{dtw}(G) \geq n$.

Moreover, we can construct a valid DAG tree decomposition by placing all sources on one side in the root bag and attaching a leaf bag to each source on the other side. This shows that $\mathrm{dtw}(G) = n$, matching the lower bound.

\end{remark}

\begin{lemma}\label{lem:dtwlb}
For a given DAG $\vec{H}$ and a corresponding set of graphs $\mathcal{G}$, there exist $G_S\in \mathcal{G}$ such that $dtw(\vec{H}) \geq \frac{tw(G_S)+1}{2}.$
\end{lemma}

\begin{proof}
Let $\vec{H}$ be a fixed acyclic orientation, and let $T = (\mathcal{B}, E)$ be a DAG tree decomposition (DTD) of $\vec{H}$. Let $G_S$ be the graph obtained from $\vec{H}$ using the construction described in ~\Cref{sec:construction}. Note that the vertex set of $G_S$ consists of the source vertices of $\vec{H}$. Since each bag in $T$ contains only source vertices and the union of all bags equals $V(G_S)$, the first condition of a tree decomposition for $G_S$ is satisfied.

To complete the construction of $G_S$, we perform a sequence of edge contractions on $\textsc{Bip}(\vec{H})$, the bipartite graph built from $\vec{H}$ as described earlier. The contraction process is guided by the structure of $T$ and proceeds in a bottom-up fashion.

For each leaf node $u \in B_{\text{leaf}}$ in $T$ and its parent $v = \textsc{Parent}(u)$, if $R(u) \cap R(v) \neq \emptyset$, we contract the non-source vertices in this intersection with $v$. In the next step, we consider the grandparent bag $B = \textsc{Parent}(v)$ and perform edge contractions for non-source vertices in $(R(B) \cap R(v)) \setminus (R(B) \cap R(u)),$ as contractions involving vertices in $R(B) \cap R(u)$ have already been completed. This recursive process continues until all contractions are completed. We now claim the following:

\begin{claim}\label{claim:edge}
   If $u$ and $v$ belong to the same bag or to adjacent bags in $T$, then $\{u, v\} \in E(G_S)$. 
\end{claim}

\begin{proof}
Suppose, for contradiction, that $\{u, v\} \in E(G_S)$, but $u$ and $v$ do not appear in the same or adjacent bags in $T$. Let $B_1$ and $B_2$ be bags containing $u$ and $v$, respectively, and consider the unique path $T(B_1, B_2)$. Let $B \in T(B_1, B_2)$ be a bag that contains neither $u$ nor $v$. Since $T$ is a valid DAG tree decomposition, it satisfies the reachability condition, implying that there exists a minimal set of source vertices $\{u_1, \ldots, u_k\} \subseteq B$ such that:

$$R(u) \cap R(v) \subseteq \bigcup_{i=1}^k \left(R(u) \cap R(u_i)\right).$$

Due to edge contraction, each $u_i$ contracts with the non-source vertices in its reachability set. Consequently, $\{u, u_i\}$ and $\{v, u_i\}$ are edges in $G_S$, implying that $\{u, v\} \notin E(G_S)$, a contradiction.
\end{proof}

Next, we verify the \emph{connectivity condition} of the tree decomposition. Suppose a source vertex $u$ appears in two bags $B_1$ and $B_2$, but not in any bag along the path $T(B_1, B_2) \setminus \{B_1, B_2\}$. From the reachability property of $T$, we know:

$$R(B_1) \cap R(B_2) \subseteq R(B),$$

for any intermediate bag $B$. Since $R(u) \subseteq R(B_1) \cap R(B_2)$, we must have $R(u) \subseteq R(B)$ as well. Therefore, there exists a minimal set of source vertices $\{u_1, \ldots, u_k\}$ such that:

$$R(u) \subseteq \bigcup_{i=1}^k R(u_i).$$

This implies that for every non-source vertex $w \in R(u)$, there exists some $u_i \neq u$ such that $w \in R(u_i)$, allowing the contraction of $w$ with $u_i$. As a result, all of $ u$'s neighbors in $G_S$ are contained in $\{u_1, \ldots, u_k\}$, and so $u$ can safely be removed from $B_2$, maintaining the connectivity condition.

We now verify the \emph{second condition} of the tree decomposition for $G_S$, namely, that for every edge $\{u, v\} \in E(G_S)$, there exists a bag that contains both $u$ and $v$. From ~\Cref{claim:edge}, such edges arise only when $u$ and $v$ are in the same or adjacent bags in $T$. We construct a new tree decomposition $T'$ from $T$ by adding for each bag $B \in T$ its parent $\textsc{Parent}(B)$ as an additional bag in $T'$. Let $\textsc{Root}(T')=\textsc{Root}(T)$, this construction guarantees that for every edge $\{u, v\} \in E(G_S)$, there exists a bag in $T'$ containing both vertices.

Let $k$ be the maximum bag size in $T$. Then each new bag in $T'$ has size at most $2k$, and hence $\mathrm{tw}(G_S) \leq 2k - 1.$ Therefore, $k \geq \frac{\mathrm{tw}(G_S) + 1}{2}.$ Since $k$ is the width of the DAG-tree decomposition $T$, this completes the proof.
\end{proof}

\begin{remark}
   
The above bound is tight. Consider the graph $G$ obtained by subdividing every edge of $K_5$. We orient all original vertices to point outward, making them sources. For each pair of original vertices, there exists a unique non-source vertex reachable from both original vertices. By  ~\Cref{obs:uniqueintersection}, such sources must appear together in the same or adjacent bags of any valid DAG tree decomposition. Therefore, $\mathrm{dtw}(G) = 3$, while $\mathrm{tw}(G) = 5$.

\end{remark}




\end{document}